\documentclass{article}
\usepackage{graphicx}
\usepackage{mathptmx} 
\usepackage[utf8]{inputenc}
\usepackage{stmaryrd}
\usepackage{amsfonts}
\usepackage{amssymb}
\usepackage{amsthm}
\usepackage{mathrsfs}
\usepackage{multirow}
\usepackage{vmargin}
\usepackage{color}
\usepackage{ifthen}
\usepackage{tikz}
\usepackage[colorlinks=true]{hyperref}

\newtheorem{theorem}{Theorem}
\newtheorem{lemma}{Lemma}
\newtheorem{problem}{Problem}

\newtheorem{corollary}{Corollary}
\newtheorem{myclaim}{Claim}

\newcommand{\ival}[2]{\llbracket #1, #2\rrbracket}
\newcommand{\xch}{\textrm{XCH}}
\newcommand{\lic}{\textrm{LIC}}

\newcommand{\nat}{\mathbb{N}}
\newcommand{\natp}{\mathbb{N}^{+}}
\newcommand{\rels}{\mathbb{N}^{*}}
\newcommand{\rel}{\mathbb{Z}}
\newcommand{\reel}{\mathbb{R}}
\newcommand{\reelp}{\mathbb{R}^{+}}
\newcommand{\reels}{\mathbb{R}^{+}}

\newlength{\legendelength}
\newlength{\legendemaxlength}
\newcommand{\llegende}[2]{
\settowidth{\legendelength}{\small\emph{#2}}
\ifthenelse{\lengthtest{\legendelength<\legendemaxlength}}
{\caption[#1]{\small \emph{#2}}}
{\caption[#1]{\small \protect\parbox[t]{0.8\textwidth}{\emph{#2}}}}
}
\newcommand{\legende}[1]{\llegende{}{#1}}

\begin{document}

\title{The hardness of routing two pairs on one face}
\author{Guyslain Naves}
% \institute{Laboratoire G-SCOP \at
%               46, avenue F\'elix Viallet 38031 GRENOBLE Cedex \\
%               Tel.: +33 (0) 476.57.49.85\\
%               \email{guyslain.naves@g-scop.fr}           %  \\
% }

\date{Received: date / Accepted: date}

\maketitle

\begin{abstract}
We prove the NP-completeness of the integer multiflow problem in planar graphs, with the following restrictions: there are only two demand edges, both lying on the infinite face of the routing graph. This was one of the open challenges concerning disjoint paths, explicitly asked by M\"{u}ller~\cite{Muller}. It also strengthens Schw\"arzler's recent proof of one of the open problems of Schrijver's book~\cite{Schrijverbook}, about the complexity of the edge-disjoint paths problem with terminals on the outer boundary of a planar graph. We also give a directed acyclic reduction. This proves that the arc-disjoint paths problem is NP-complete in directed acyclic graphs, even with only two demand arcs.
%\keywords{integer multiflows \and Planar graphs \and Disjoint Paths}
%\subclass{05C38}
\end{abstract}

\section{Introduction}

The multiflow problem has been studied in combinatorial optimization for many years, both because of its theoretical interests and applications. Basically, we want to find integer flows between pairs of terminals, respecting capacity constraints.\\
The general problem is NP-complete, with different types of constraints, see e.g. the survey of Frank~\cite{Frank}. The most general way to define constraints is to put capacities on the edges of the graph, in the same way as for the classical flow problem. When these capacities are $1$ everywhere, this defines the edge-disjoint (or arc-disjoint) paths problem. Robertson and Seymour~\cite{Robertson95} proved that the multiflow problem is polynomial for undirected graphs, assuming that the total demand is fixed.\\ 
Special interest has been shown for solving the problem in planar graphs (directed or not). Kramer and Van Leeuwen~\cite{Kramer84} have shown that the undirected planar multiflow problem is NP-complete in the general case. Nevertheless, a good characterization theorem has been proved by Okamura and Seymour~\cite{Okamura} for the edge-disjoint paths problem in planar Eulerian graphs, under the assumption that all terminals are on the boundary of a unique face of the graph. Despite sharpenings (see~\cite{Frank85},~\cite{Okamura83}), the non-Eulerian case remained open until 2007, when Schw\"{a}rzler~\cite{Schwarzler} proved the NP-completeness of the edge-disjoint paths problem in planar graphs with all terminals on the boundary of the same face of the graph.\\
Between Robertson and Seymour's result, and Schw\"{a}rzler's result, one could ask if there is a polynomial-time algorithm for the edge-disjoint paths problem in planar graphs, when the number of different pairs of terminals is fixed, but each may be repeated many times. Especially when there are only two parallel classes of demands, with terminals lying on a single face of the graph. Actually, Schw\"{a}rzler's proof can easily be modified to prove that the problem is still NP-complete with three pairs of terminals. In this paper we give a new reduction, proving that the problem is NP-complete with only two pairs of terminals. This solves a question of M\"{u}ller~\cite{Muller}.\\
We will also give a directed version of our proof, showing that the arc-disjoint paths problem is NP-complete, even if $G$ is planar with two opposite parallel classes of demand edges $st$ and $ts$ where vertices $s$ and $t$ belong to the boundary of the same face of $G$. Both results strengthen~\cite{Muller}. Finally, we prove the NP-completeness when $G$ is a planar acyclic digraph and $H$ consists of two pairs of terminals lying on the outer face of $G$.\\

\section{Definitions}\label{sec:def}

Let $G=(V,E)$ be an undirected graph, and let $c : E \to \nat$ be a \emph{capacity function} on the edges of $G$. Let $H=(T,D)$ be an undirected graph with $T \subseteq V$, and $r : D \to \nat$ a \emph{demand function} (or \emph{request}). The \emph{multiflow problem} is to find a multiset $\mathcal{C}$ of cycles of $G+H$ satisfying the following conditions :
\begin{description}
\item[$(i)$] Each cycle of $\mathcal{C}$ contains exactly one edge of $H$.
\item[$(ii)$] For each edge of $G$, the number of cycles in $\mathcal{C}$ using it is less than its capacity.
\item[$(iii)$] For each edge of $H$, the number of cycles in $\mathcal{C}$ using it is exactly its request.
\end{description}
$H$ is usually called the \emph{demand graph}, $T$ is the set of \emph{terminals}. By cycle, we mean a closed sequence of distinct edges that are consecutive in the graph, or equivalently a connected Eulerian subgraph. The problem can easily be defined in digraphs, by replacing every occurrence of ``cycle'' by ``directed cycle''. Thus, an instance of the multiflow problem consists of a quadruple $(G,H,r,c)$. In the following, $c$ will always be supposed to be equal to $1$. We will note $\mathcal{P}$ the set of paths obtained from $\mathcal{C}$ by ignoring the demand edges, and we will mainly speak about these paths instead of the cycles.\\
A \emph{path} is formally a sequence of distinct edges that are consecutive. Thus, cycles are closed paths. Two paths are \emph{edge-disjoint} if their edge-sets are disjoint. We define \emph{directed paths} and \emph{arc-disjoint paths} analogously. When $c$ is constantly equal to $1$, the multiflow problem is known as the \emph{edge-disjoint paths problem} (respectively the \emph{arc-disjoint paths problem}).\\
The graphs considered in this paper are always without loops, but parallel edges are allowed. Actually, whenever an edge $e \in E(G)$ has a capacity greater than $1$, we replace it by $c(e)$ parallel edges. Let $U \in V$ be a subset of a vertex set of the graph. We note $\delta(U)$ the set of edges having exactly one extremity in $U$. Every set of edges that can be written as $\delta(U)$ for some $U$ is called a \emph{cut} of the graph. In directed graphs, $\delta^-(U)$ is the set of arcs entering $U$, $\delta^+(U)$ is the set of arcs leaving $U$. When $\delta(U) = \delta^+(U)$, we say that $U$ is a \emph{directed cut}. We define $d(U) := |\delta(U)|$ and similarly $d^+(U)$ and $d^-(U)$.\\
Let $C$ be a cut of $G+H$. $C$ is a \emph{tight cut} if $c(C \cap E) - r(C \cap D) = 0$. If this difference is negative, the multiflow problem is not feasible. When $H$ is reduced to a single edge, the famous \emph{max-flow-min-cut} result states that a multiflow exists if and only if this difference is never negative (see Menger's theorem~\cite{Menger}). When $C$ is a tight cut, each edge of $C$ is entirely used in any solution of the multiflow problem : there are as many paths through each edge as its capacity. In directed graphs, tight cuts are the cuts $\delta(U)$ of $G+H$ with $c(\delta^+_G(U)) - r(\delta^-_H(U)) = 0$.\\
A \emph{planar graph} is a graph that has an embedding in the plane without intersection of the edges (or arcs). Let $P_1$ and $P_2$ be two edge-disjoint paths in $G$. $P_1$ and $P_2$ \emph{cross} at vertex $v \in V(G)$ if there are four edges $e_1,\ldots,e_4$ incident to $v$, appearing in this order around $v$, such that $e_1$ and $e_3$ are consecutive in $P_1$ and $e_2$, $e_4$ are consecutive in $P_2$. However, two paths may have a common vertex without crossing, see Figure~\ref{fig:biflot-croise}. A \emph{crossing} is a triple $(P,Q,u)$ such that $P$ and $Q$ crosses at vertex $u$.\\

\begin{figure}
\begin{center}
\begin{tikzpicture}[x=0.5cm,y=0.5cm,>=latex]
% Fichier d'origine : th-biflot-croise.pgf
\node[circle,inner sep = 0pt,minimum size =3pt,fill = black] (nx15y17) at (15,17) {};
\node[circle,inner sep = 0pt,minimum size =3pt,fill = black] (nx8y17) at (8,17) {};
\node[circle,inner sep = 0pt,minimum size =0pt,fill = black] (nx17y19) at (17,19) {};
\node[circle,inner sep = 0pt,minimum size =0pt,fill = black] (nx13y15) at (13,15) {};
\node[circle,inner sep = 0pt,minimum size =0pt,fill = black] (nx17y15) at (17,15) {};
\node[circle,inner sep = 0pt,minimum size =0pt,fill = black] (nx13y19) at (13,19) {};
\node[circle,inner sep = 0pt,minimum size =0pt,fill = black] (nx10y19) at (10,19) {};
\node[circle,inner sep = 0pt,minimum size =0pt,fill = black] (nx10y15) at (10,15) {};
\node[circle,inner sep = 0pt,minimum size =0pt,fill = black] (nx6y19) at (6,19) {};
\node[circle,inner sep = 0pt,minimum size =0pt,fill = black] (nx6y15) at (6,15) {};
\draw[dashed] (nx13y15) -- (15,17) -- (nx17y19);
\draw[dotted] (nx13y19) -- (15,17) -- (nx17y15);
\draw[dashed, rounded corners=6pt] (nx10y15) -- (8,17) -- (nx10y19);
\draw[dotted, rounded corners=6pt] (nx6y15) -- (8,17) -- (nx6y19);
\end{tikzpicture}
\end{center}
\legende{The two paths on the left do not cross, those on the right cross each other.}
\label{fig:biflot-croise}
\end{figure}

\section{Outline of the proof}\label{sec:outline}

We will prove the NP-completeness of the edge-disjoint paths problem with only two pairs of terminals, by reduction from {\sc 3-Sat}. Before giving the full proof, which is quite technical, we explain the main ideas of the reduction.\\
Consider an instance of {\sc 3-Sat}, consisting of clauses over a set of variables. We build a graph in the form of a grid, with as many columns as there are clauses, and as many rows as twice the number of variables. In this grid, the intersection of each row with each column is one of two basic graphs. These two graphs, called $\xch$ and $\lic$ should have the following properties. They have two vertices of degree one in their left and right sides, and four --- actually two pairs --- at their top and their bottom. We suppose that there will always be two paths either from the two left top vertices or from the two right ones, going through them from top to bottom (\emph{vertical paths}). We will also always route one or two paths from left to right (\emph{horizontal paths}). If there are two horizontal paths, then the two vertical paths can only be routed diagonally. The behaviours of $\xch$ and $\lic$ are different if there is only one horizontal paths. In that case, in $\xch$, the vertical paths are still forced to go diagonally through the graph, but in $\lic$, they can be routed vertically. Figures~\ref{fig:biflot-pres1} and~\ref{fig:biflot-pres2} illustrate the possible routing through the two graphs (and what we mean by \emph{diagonally} and \emph{vertically}).

\begin{figure}[htb]
\begin{center}
\begin{tikzpicture}[x=0.3cm,y=0.3cm,>=latex]
% Fichier d'origine : th-biflot-pres1.pgf
\definecolor{lightg}{rgb}{0.8,0.8,0.8}
\fill[lightg,rounded corners=2pt] (17.5,6.5) rectangle (27.5,11.5);
\fill[lightg,rounded corners=2pt] (2.5,6.5) rectangle (12.5,11.5);
\fill[lightg,rounded corners=2pt] (17.5,17.5) rectangle (27.5,22.5);
\fill[lightg,rounded corners=2pt] (2.5,17.5) rectangle (12.5,22.5);
\node[circle,inner sep = 0pt,minimum size =5pt,fill = black] (nx11y13) at (11,13) {};
\node[circle,inner sep = 0pt,minimum size =5pt,fill = black] (nx9y13) at (9,13) {};
\node[circle,inner sep = 0pt,minimum size =5pt,fill = black] (nx6y13) at (6,13) {};
\node[circle,inner sep = 0pt,minimum size =5pt,fill = black] (nx4y13) at (4,13) {};
\node[circle,inner sep = 0pt,minimum size =5pt,fill = black] (nx1y10) at (1,10) {};
\node[circle,inner sep = 0pt,minimum size =5pt,fill = black] (nx1y8) at (1,8) {};
\node[circle,inner sep = 0pt,minimum size =5pt,fill = black] (nx4y5) at (4,5) {};
\node[circle,inner sep = 0pt,minimum size =5pt,fill = black] (nx6y5) at (6,5) {};
\node[circle,inner sep = 0pt,minimum size =5pt,fill = black] (nx9y5) at (9,5) {};
\node[circle,inner sep = 0pt,minimum size =5pt,fill = black] (nx11y5) at (11,5) {};
\node[circle,inner sep = 0pt,minimum size =5pt,fill = black] (nx14y8) at (14,8) {};
\node[circle,inner sep = 0pt,minimum size =5pt,fill = black] (nx14y10) at (14,10) {};
\node[circle,inner sep = 0pt,minimum size =5pt,fill = black] (nx16y10) at (16,10) {};
\node[circle,inner sep = 0pt,minimum size =5pt,fill = black] (nx16y8) at (16,8) {};
\node[circle,inner sep = 0pt,minimum size =5pt,fill = black] (nx19y5) at (19,5) {};
\node[circle,inner sep = 0pt,minimum size =5pt,fill = black] (nx21y5) at (21,5) {};
\node[circle,inner sep = 0pt,minimum size =5pt,fill = black] (nx24y5) at (24,5) {};
\node[circle,inner sep = 0pt,minimum size =5pt,fill = black] (nx26y5) at (26,5) {};
\node[circle,inner sep = 0pt,minimum size =5pt,fill = black] (nx29y8) at (29,8) {};
\node[circle,inner sep = 0pt,minimum size =5pt,fill = black] (nx29y10) at (29,10) {};
\node[circle,inner sep = 0pt,minimum size =5pt,fill = black] (nx26y13) at (26,13) {};
\node[circle,inner sep = 0pt,minimum size =5pt,fill = black] (nx24y13) at (24,13) {};
\node[circle,inner sep = 0pt,minimum size =5pt,fill = black] (nx21y13) at (21,13) {};
\node[circle,inner sep = 0pt,minimum size =5pt,fill = black] (nx19y13) at (19,13) {};
\node[circle,inner sep = 0pt,minimum size =5pt,fill = black] (nx16y21) at (16,21) {};
\node[circle,inner sep = 0pt,minimum size =5pt,fill = black] (nx16y19) at (16,19) {};
\node[circle,inner sep = 0pt,minimum size =5pt,fill = black] (nx19y16) at (19,16) {};
\node[circle,inner sep = 0pt,minimum size =5pt,fill = black] (nx21y16) at (21,16) {};
\node[circle,inner sep = 0pt,minimum size =5pt,fill = black] (nx24y16) at (24,16) {};
\node[circle,inner sep = 0pt,minimum size =5pt,fill = black] (nx26y16) at (26,16) {};
\node[circle,inner sep = 0pt,minimum size =5pt,fill = black] (nx29y19) at (29,19) {};
\node[circle,inner sep = 0pt,minimum size =5pt,fill = black] (nx29y21) at (29,21) {};
\node[circle,inner sep = 0pt,minimum size =5pt,fill = black] (nx26y24) at (26,24) {};
\node[circle,inner sep = 0pt,minimum size =5pt,fill = black] (nx24y24) at (24,24) {};
\node[circle,inner sep = 0pt,minimum size =5pt,fill = black] (nx21y24) at (21,24) {};
\node[circle,inner sep = 0pt,minimum size =5pt,fill = black] (nx19y24) at (19,24) {};
\node[circle,inner sep = 0pt,minimum size =5pt,fill = black] (nx4y24) at (4,24) {};
\node[circle,inner sep = 0pt,minimum size =5pt,fill = black] (nx6y24) at (6,24) {};
\node[circle,inner sep = 0pt,minimum size =5pt,fill = black] (nx9y24) at (9,24) {};
\node[circle,inner sep = 0pt,minimum size =5pt,fill = black] (nx11y24) at (11,24) {};
\node[circle,inner sep = 0pt,minimum size =5pt,fill = black] (nx14y21) at (14,21) {};
\node[circle,inner sep = 0pt,minimum size =5pt,fill = black] (nx14y19) at (14,19) {};
\node[circle,inner sep = 0pt,minimum size =5pt,fill = black] (nx11y16) at (11,16) {};
\node[circle,inner sep = 0pt,minimum size =5pt,fill = black] (nx9y16) at (9,16) {};
\node[circle,inner sep = 0pt,minimum size =5pt,fill = black] (nx6y16) at (6,16) {};
\node[circle,inner sep = 0pt,minimum size =5pt,fill = black] (nx4y16) at (4,16) {};
\node[circle,inner sep = 0pt,minimum size =5pt,fill = black] (nx1y19) at (1,19) {};
\node[circle,inner sep = 0pt,minimum size =5pt,fill = black] (nx1y21) at (1,21) {};
\node[circle,inner sep = 0pt,minimum size =0pt,fill = black] (nx19y11) at (19,11) {};
\node[circle,inner sep = 0pt,minimum size =0pt,fill = black] (nx21y11) at (21,11) {};
\node[circle,inner sep = 0pt,minimum size =0pt,fill = black] (nx24y11) at (24,11) {};
\node[circle,inner sep = 0pt,minimum size =0pt,fill = black] (nx26y11) at (26,11) {};
\node[circle,inner sep = 0pt,minimum size =0pt,fill = black] (nx27y10) at (27,10) {};
\node[circle,inner sep = 0pt,minimum size =0pt,fill = black] (nx27y8) at (27,8) {};
\node[circle,inner sep = 0pt,minimum size =0pt,fill = black] (nx26y7) at (26,7) {};
\node[circle,inner sep = 0pt,minimum size =0pt,fill = black] (nx24y7) at (24,7) {};
\node[circle,inner sep = 0pt,minimum size =0pt,fill = black] (nx21y7) at (21,7) {};
\node[circle,inner sep = 0pt,minimum size =0pt,fill = black] (nx19y7) at (19,7) {};
\node[circle,inner sep = 0pt,minimum size =0pt,fill = black] (nx18y8) at (18,8) {};
\node[circle,inner sep = 0pt,minimum size =0pt,fill = black] (nx18y10) at (18,10) {};
\node[circle,inner sep = 0pt,minimum size =0pt,fill = black] (nx3y10) at (3,10) {};
\node[circle,inner sep = 0pt,minimum size =0pt,fill = black] (nx3y8) at (3,8) {};
\node[circle,inner sep = 0pt,minimum size =0pt,fill = black] (nx4y7) at (4,7) {};
\node[circle,inner sep = 0pt,minimum size =0pt,fill = black] (nx6y7) at (6,7) {};
\node[circle,inner sep = 0pt,minimum size =0pt,fill = black] (nx9y7) at (9,7) {};
\node[circle,inner sep = 0pt,minimum size =0pt,fill = black] (nx11y7) at (11,7) {};
\node[circle,inner sep = 0pt,minimum size =0pt,fill = black] (nx12y8) at (12,8) {};
\node[circle,inner sep = 0pt,minimum size =0pt,fill = black] (nx12y10) at (12,10) {};
\node[circle,inner sep = 0pt,minimum size =0pt,fill = black] (nx11y11) at (11,11) {};
\node[circle,inner sep = 0pt,minimum size =0pt,fill = black] (nx9y11) at (9,11) {};
\node[circle,inner sep = 0pt,minimum size =0pt,fill = black] (nx6y11) at (6,11) {};
\node[circle,inner sep = 0pt,minimum size =0pt,fill = black] (nx4y11) at (4,11) {};
\node[circle,inner sep = 0pt,minimum size =0pt,fill = black] (nx18y19) at (18,19) {};
\node[circle,inner sep = 0pt,minimum size =0pt,fill = black] (nx19y18) at (19,18) {};
\node[circle,inner sep = 0pt,minimum size =0pt,fill = black] (nx21y18) at (21,18) {};
\node[circle,inner sep = 0pt,minimum size =0pt,fill = black] (nx24y18) at (24,18) {};
\node[circle,inner sep = 0pt,minimum size =0pt,fill = black] (nx26y18) at (26,18) {};
\node[circle,inner sep = 0pt,minimum size =0pt,fill = black] (nx27y19) at (27,19) {};
\node[circle,inner sep = 0pt,minimum size =0pt,fill = black] (nx27y21) at (27,21) {};
\node[circle,inner sep = 0pt,minimum size =0pt,fill = black] (nx26y22) at (26,22) {};
\node[circle,inner sep = 0pt,minimum size =0pt,fill = black] (nx24y22) at (24,22) {};
\node[circle,inner sep = 0pt,minimum size =0pt,fill = black] (nx21y22) at (21,22) {};
\node[circle,inner sep = 0pt,minimum size =0pt,fill = black] (nx19y22) at (19,22) {};
\node[circle,inner sep = 0pt,minimum size =0pt,fill = black] (nx18y21) at (18,21) {};
\node[circle,inner sep = 0pt,minimum size =0pt,fill = black] (nx6y22) at (6,22) {};
\node[circle,inner sep = 0pt,minimum size =0pt,fill = black] (nx9y22) at (9,22) {};
\node[circle,inner sep = 0pt,minimum size =0pt,fill = black] (nx11y22) at (11,22) {};
\node[circle,inner sep = 0pt,minimum size =0pt,fill = black] (nx12y21) at (12,21) {};
\node[circle,inner sep = 0pt,minimum size =0pt,fill = black] (nx12y19) at (12,19) {};
\node[circle,inner sep = 0pt,minimum size =0pt,fill = black] (nx11y18) at (11,18) {};
\node[circle,inner sep = 0pt,minimum size =0pt,fill = black] (nx9y18) at (9,18) {};
\node[circle,inner sep = 0pt,minimum size =0pt,fill = black] (nx6y18) at (6,18) {};
\node[circle,inner sep = 0pt,minimum size =0pt,fill = black] (nx4y18) at (4,18) {};
\node[circle,inner sep = 0pt,minimum size =0pt,fill = black] (nx3y19) at (3,19) {};
\node[circle,inner sep = 0pt,minimum size =0pt,fill = black] (nx3y21) at (3,21) {};
\node[circle,inner sep = 0pt,minimum size =0pt,fill = black] (nx4y22) at (4,22) {};
\draw[black] (nx21y7) -- (nx26y11);
\draw[black] (nx24y11) -- (nx19y7);
\draw[black] (nx18y10) -- (nx27y8);
\draw[black] (nx19y13) -- (nx19y11);
\draw[black] (nx21y11) -- (nx21y13);
\draw[black] (nx24y13) -- (nx24y11);
\draw[black] (nx26y11) -- (nx26y13);
\draw[black] (nx29y10) -- (nx27y10);
\draw[black] (nx27y8) -- (nx29y8);
\draw[black] (nx26y5) -- (nx26y7);
\draw[black] (nx24y7) -- (nx24y5);
\draw[black] (nx21y5) -- (nx21y7);
\draw[black] (nx19y7) -- (nx19y5);
\draw[black] (nx18y8) -- (nx16y8);
\draw[black] (nx16y10) -- (nx18y10);
\draw[black] (nx6y11) -- (nx11y7);
\draw[black] (nx4y11) -- (nx9y7);
\draw[black] (nx3y10) -- (nx12y10);
\draw[black] (nx1y10) -- (nx3y10);
\draw[black] (nx3y8) -- (nx1y8);
\draw[black] (nx4y5) -- (nx4y7);
\draw[black] (nx6y7) -- (nx6y5);
\draw[black] (nx9y5) -- (nx9y7);
\draw[black] (nx11y7) -- (nx11y5);
\draw[black] (nx14y8) -- (nx12y8);
\draw[black] (nx12y10) -- (nx14y10);
\draw[black] (nx11y13) -- (nx11y11);
\draw[black] (nx9y11) -- (nx9y13);
\draw[black] (nx6y13) -- (nx6y11);
\draw[black] (nx4y13) -- (nx4y11);
\draw[black] (nx27y19) -- (nx18y19);
\draw[black] (nx18y21) -- (nx27y21);
\draw[black] (nx26y22) -- (nx21y18);
\draw[black] (nx24y22) -- (nx19y18);
\draw[black] (nx18y19) -- (nx16y19);
\draw[black] (nx19y16) -- (nx19y18);
\draw[black] (nx21y18) -- (nx21y16);
\draw[black] (nx24y16) -- (nx24y18);
\draw[black] (nx26y18) -- (nx26y16);
\draw[black] (nx29y19) -- (nx27y19);
\draw[black] (nx27y21) -- (nx29y21);
\draw[black] (nx26y24) -- (nx26y22);
\draw[black] (nx24y22) -- (nx24y24);
\draw[black] (nx21y24) -- (nx21y22);
\draw[black] (nx19y22) -- (nx19y24);
\draw[black] (nx16y21) -- (nx18y21);
\draw[black] (nx6y22) -- (nx11y18);
\draw[black] (nx4y22) -- (nx9y18);
\draw[black] (nx12y19) -- (nx3y19);
\draw[black] (nx3y21) -- (nx12y21);
\draw[black] (nx6y22) -- (nx6y24);
\draw[black] (nx9y22) -- (nx9y24);
\draw[black] (nx11y22) -- (nx11y24);
\draw[black] (nx14y21) -- (nx12y21);
\draw[black] (nx12y19) -- (nx14y19);
\draw[black] (nx11y16) -- (nx11y18);
\draw[black] (nx9y18) -- (nx9y16);
\draw[black] (nx6y16) -- (nx6y18);
\draw[black] (nx4y18) -- (nx4y16);
\draw[black] (nx1y19) -- (nx3y19);
\draw[black] (nx3y21) -- (nx1y21);
\draw[black] (nx4y24) -- (nx4y22);
\end{tikzpicture}
\end{center}
\legende{Possible routings through $\xch$ and $\lic$.}
\label{fig:biflot-pres1}
\end{figure}

\begin{figure}[htb]
\begin{center}
\begin{tikzpicture}[x=0.3cm,y=0.3cm,>=latex]
% Fichier d'origine : th-biflot-pres2.pgf
\definecolor{lightg}{rgb}{0.8,0.8,0.8}
\fill[lightg,rounded corners=2pt] (17.5,17.5) rectangle (27.5,22.5);
\fill[lightg,rounded corners=2pt] (2.5,17.5) rectangle (12.5,22.5);
\node[circle,inner sep = 0pt,minimum size =5pt,fill = black] (nx1y21) at (1,21) {};
\node[circle,inner sep = 0pt,minimum size =5pt,fill = black] (nx1y19) at (1,19) {};
\node[circle,inner sep = 0pt,minimum size =5pt,fill = black] (nx4y16) at (4,16) {};
\node[circle,inner sep = 0pt,minimum size =5pt,fill = black] (nx6y16) at (6,16) {};
\node[circle,inner sep = 0pt,minimum size =5pt,fill = black] (nx9y16) at (9,16) {};
\node[circle,inner sep = 0pt,minimum size =5pt,fill = black] (nx11y16) at (11,16) {};
\node[circle,inner sep = 0pt,minimum size =5pt,fill = black] (nx14y19) at (14,19) {};
\node[circle,inner sep = 0pt,minimum size =5pt,fill = black] (nx14y21) at (14,21) {};
\node[circle,inner sep = 0pt,minimum size =5pt,fill = black] (nx16y21) at (16,21) {};
\node[circle,inner sep = 0pt,minimum size =5pt,fill = black] (nx16y19) at (16,19) {};
\node[circle,inner sep = 0pt,minimum size =5pt,fill = black] (nx19y16) at (19,16) {};
\node[circle,inner sep = 0pt,minimum size =5pt,fill = black] (nx21y16) at (21,16) {};
\node[circle,inner sep = 0pt,minimum size =5pt,fill = black] (nx24y16) at (24,16) {};
\node[circle,inner sep = 0pt,minimum size =5pt,fill = black] (nx26y16) at (26,16) {};
\node[circle,inner sep = 0pt,minimum size =5pt,fill = black] (nx29y19) at (29,19) {};
\node[circle,inner sep = 0pt,minimum size =5pt,fill = black] (nx29y21) at (29,21) {};
\node[circle,inner sep = 0pt,minimum size =5pt,fill = black] (nx26y24) at (26,24) {};
\node[circle,inner sep = 0pt,minimum size =5pt,fill = black] (nx24y24) at (24,24) {};
\node[circle,inner sep = 0pt,minimum size =5pt,fill = black] (nx21y24) at (21,24) {};
\node[circle,inner sep = 0pt,minimum size =5pt,fill = black] (nx19y24) at (19,24) {};
\node[circle,inner sep = 0pt,minimum size =5pt,fill = black] (nx11y24) at (11,24) {};
\node[circle,inner sep = 0pt,minimum size =5pt,fill = black] (nx9y24) at (9,24) {};
\node[circle,inner sep = 0pt,minimum size =5pt,fill = black] (nx6y24) at (6,24) {};
\node[circle,inner sep = 0pt,minimum size =5pt,fill = black] (nx4y24) at (4,24) {};
\node[circle,inner sep = 0pt,minimum size =0pt,fill = black] (nx3y21) at (3,21) {};
\node[circle,inner sep = 0pt,minimum size =0pt,fill = black] (nx4y22) at (4,22) {};
\node[circle,inner sep = 0pt,minimum size =0pt,fill = black] (nx6y22) at (6,22) {};
\node[circle,inner sep = 0pt,minimum size =0pt,fill = black] (nx9y22) at (9,22) {};
\node[circle,inner sep = 0pt,minimum size =0pt,fill = black] (nx11y22) at (11,22) {};
\node[circle,inner sep = 0pt,minimum size =0pt,fill = black] (nx12y21) at (12,21) {};
\node[circle,inner sep = 0pt,minimum size =0pt,fill = black] (nx12y19) at (12,19) {};
\node[circle,inner sep = 0pt,minimum size =0pt,fill = black] (nx11y18) at (11,18) {};
\node[circle,inner sep = 0pt,minimum size =0pt,fill = black] (nx9y18) at (9,18) {};
\node[circle,inner sep = 0pt,minimum size =0pt,fill = black] (nx6y18) at (6,18) {};
\node[circle,inner sep = 0pt,minimum size =0pt,fill = black] (nx4y18) at (4,18) {};
\node[circle,inner sep = 0pt,minimum size =0pt,fill = black] (nx3y19) at (3,19) {};
\node[circle,inner sep = 0pt,minimum size =0pt,fill = black] (nx18y21) at (18,21) {};
\node[circle,inner sep = 0pt,minimum size =0pt,fill = black] (nx19y22) at (19,22) {};
\node[circle,inner sep = 0pt,minimum size =0pt,fill = black] (nx21y22) at (21,22) {};
\node[circle,inner sep = 0pt,minimum size =0pt,fill = black] (nx24y22) at (24,22) {};
\node[circle,inner sep = 0pt,minimum size =0pt,fill = black] (nx26y22) at (26,22) {};
\node[circle,inner sep = 0pt,minimum size =0pt,fill = black] (nx27y21) at (27,21) {};
\node[circle,inner sep = 0pt,minimum size =0pt,fill = black] (nx27y19) at (27,19) {};
\node[circle,inner sep = 0pt,minimum size =0pt,fill = black] (nx26y18) at (26,18) {};
\node[circle,inner sep = 0pt,minimum size =0pt,fill = black] (nx24y18) at (24,18) {};
\node[circle,inner sep = 0pt,minimum size =0pt,fill = black] (nx21y18) at (21,18) {};
\node[circle,inner sep = 0pt,minimum size =0pt,fill = black] (nx19y18) at (19,18) {};
\node[circle,inner sep = 0pt,minimum size =0pt,fill = black] (nx18y19) at (18,19) {};
\draw[black] (nx1y21) -- (nx3y21);
\draw[black] (nx4y22) -- (nx4y24);
\draw[black] (nx6y24) -- (nx6y22);
\draw[black] (nx9y22) -- (nx9y24);
\draw[black] (nx11y24) -- (nx11y22);
\draw[black] (nx12y21) -- (nx14y21);
\draw[black] (nx14y19) -- (nx12y19);
\draw[black] (nx11y18) -- (nx11y16);
\draw[black] (nx9y16) -- (nx9y18);
\draw[black] (nx6y18) -- (nx6y16);
\draw[black] (nx4y16) -- (nx4y18);
\draw[black] (nx3y19) -- (nx1y19);
\draw[black] (nx16y21) -- (nx18y21);
\draw[black] (nx19y22) -- (nx19y24);
\draw[black] (nx21y24) -- (nx21y22);
\draw[black] (nx24y22) -- (nx24y24);
\draw[black] (nx26y24) -- (nx26y22);
\draw[black] (nx27y21) -- (nx29y21);
\draw[black] (nx29y19) -- (nx27y19);
\draw[black] (nx26y18) -- (nx26y16);
\draw[black] (nx24y16) -- (nx24y18);
\draw[black] (nx21y18) -- (nx21y16);
\draw[black] (nx19y16) -- (nx19y18);
\draw[black] (nx18y19) -- (nx16y19);
\draw[black] (nx3y21) -- (nx12y19);
\draw[black] (nx4y22) -- (nx4y18);
\draw[black] (nx6y22) -- (nx6y18);
\draw[black] (nx18y19) -- (nx27y19);
\draw[black] (nx24y22) -- (nx24y18);
\draw[black] (nx26y22) -- (nx26y18);
\end{tikzpicture}
\end{center}
\legende{Additional routings for $\lic$. When there is only one horizontal path, the vertical paths are not forced to be diagonal.}
\label{fig:biflot-pres2}
\end{figure}

The typical behaviour for these gadgets is to change the vertical paths from one side to the other (Lemmas~\ref{lemme:xch-22} and~\ref{lemme:lic-22}). We say that they \emph{shift} the paths. The only special case when the two vertical paths can stay on the same side is the following : there is a $\lic$ and a single horizontal path. In that special case, the gadget \emph{keeps} the paths (Lemma~\ref{lemme:lic-12}). The reduction is basically the following : there are two consecutive rows for each variable, and we route three paths along these two rows. Thus, one path will follow one row, and the two other paths will follow the other row, \emph{deciding a variable assignment}. We also route two paths in each column, and we ask that \emph{these two paths are kept an odd number of times}. Because of the properties of $\lic$ and $\xch$, the two vertical paths can be kept in a particular gadget if and only if this gadget is a $\lic$ and there is only one horizontal path in the corresponding row. We place $\lic$ on those particular intersections for which the literal associated with the row appears in the clause encoded by the column. Thus, we ensure that whenever vertical paths are kept, the corresponding clause is made valid by the chosen assignment. Then we must guarantee that the two paths associated with each column are kept at least once (actually an odd number of times). This is done by asking these two paths to be routed from the two left pair of uppermost vertices of the column, to the two rightmost lower vertices. Because the number of rows is even (twice the number of variables), paths must indeed be kept an odd number of times.\\
Finally, we can add two common terminals for the vertical paths, and two others for the horizontal paths, achieving the desired restriction to two edges of demand. Unfortunately, there are several difficulties in implementing this reduction. Mainly, $\xch$ and $\lic$ does not exist, with the prescribed properties. Actually, the main difficulty is to enforce that the horizontal paths stay on their respective rows.  We will prove that with our gadgets, \emph{the horizontal paths cannot go through more than three rows in each column} (Lemma~\ref{lemme:decalage-nc}). Our solution is then to create a \emph{buffer consisting of a large number of rows}, between the rows encoding the variables. Similarly, we must prove that \emph{the vertical paths stay in their columns} (Lemmas~\ref{mainlemma} and~\ref{lemme:verticalpaths}). This  will be done, by using the property that our graph is \emph{nearly Eulerian} (almost each vertex has even degree). This fact will help us to prove that the edges not in a solution induce cycles and a small set of paths, called \emph{no-paths}. One of the main lemmas states that the extremities of each no-path are determined (Lemma~\ref{lemme:non-chemins}). More exactly, there will be exactly one no-path per variable. Considering horizontal paths and no-paths together, many cuts are tight for them. Consequently, vertical paths cannot use the edges contained in these cuts : each vertical path intersects only one column. Then, local properties of $\xch$ and $\lic$ allow us to conclude the proof.

\section{Preliminaries}

\subsection{Uncrossing the paths}

 We can suppose without loss of generality that each pair of paths induces at most one crossing :

\begin{lemma}\label{lemme:uncrossing}
Let $(G,H)$ be an instance of the edge-disjoint paths problem with $G$ planar. There is a solution for $(G,H)$ if and only if there is a solution such that each pair of paths crosses at most once, and two paths with the same extremities do not cross. 
\end{lemma}

\begin{proof}
Let $\mathcal{P}$ be a solution minimizing the number of crossings. We can assume that every path is simple. Suppose that two paths $P_1$ and $P_2$ induce more than one crossing. Let $u$ and $v$ be the first and second crossings between $P_1$ and $P_2$, starting from one extremity of $P_1$.\\
We decompose $P_i$ in three paths $Q_i \cup R_i \cup S_i$ where $R_i$ has extremities $u$ and $v$. Then replacing $P_1$ and $P_2$ by $P'_1 = Q_1 \cup R_2 \cup S_1$ and $P'_2 = Q_2 \cup R_1 \cup S_2$, we show that the number of crossings is reduced, contradicting the minimality of $\mathcal{P}$.\\
In every vertex except $u$ and $v$, the paths are not locally modified, thus the number of crossings is not changed. The number of crossings between $P'_1$ and $P'_2$ is reduced by at least $2$. Then, the neighbourhood of vertex $u$ (and symmetrically $v$) is divided into four parts. A path going through two consecutive parts crosses one of $P_1$ and $P_2$, and one of $P'_1$ and $P'_2$. A path going through two opposite parts crosses both $P_1$ and $P_2$. This proves the number of crossing is decreased.\\
Finally, if $P_1$ and $P_2$ have a common extremity $u$, and crosses at $v$, the same transformation applies again, decreasing the number of crossing by at least one.
\end{proof}

Solutions will always be supposed \emph{uncrossed} (each pair of paths with different extremities induces at most one crossing) and simple. The following is an easy consequence of uncrossing :

\begin{lemma}\label{lemme:uncross-biflot}
Let $G$ be a planar graph, $a$, $b$, $c$ and $d$ four vertices on the boundary of the infinite face of $G$, occuring in this order. Let $\mathcal{P}$ be an uncrossed set of $(a,c)$-paths and $(b,d)$-paths mutually edge-disjoint. Then, all $(a,c)$-paths cross the $(b,d)$-paths in the same order.
\end{lemma}

\begin{proof}
There is a crossing neither between the $(b,d)$-paths, nor between the $(a,c)$-paths, so we can choose $P$ among the $(b,d)$-paths such that all the other $(b,d)$-paths are on the same side. We prove the lemma by induction on the number of $(b,d)$-paths. $P$ can be closed to a cycle by adding a curve on the infinite face of $G$. Then by Jordan's theorem, $P$ separates the $(b,d)$-paths from $a$ (say). All $(a,c)$-paths, starting from $a$ must cross $P$ before the other $(b,d)$-paths. Using induction on the $(b,d)$-paths minus $P$ concludes the proof.
\end{proof}

\subsection{Forbidding crossing}

It happens to be useful to have in the proof an Eulerian graph. This is not possible since the edge-disjoint paths problem with terminals on the boundary of the infinite face is polynomially solvable as soon as all the inner vertices have even degrees~\cite{Okamura},~\cite{Frank85}. We introduce the following restriction : every vertex will be of degree four, but in some of them, paths will not be allowed to cross each other. This effect can easily be achieved by replacing those special vertices by a cycle of length $4$, as described in Figure~\ref{fig:biflot-ext} 

\begin{figure}
\begin{center}
\begin{tikzpicture}[x=0.5cm,y=0.5cm,>=latex]
% Fichier d'origine : th-biflot-ext.pgf
\node[circle,inner sep = 0pt,minimum size =5pt,fill = black] (nx16y17) at (16,17) {};
\node[circle,inner sep = 0pt,minimum size =5pt,fill = black] (nx15y18) at (15,18) {};
\node[circle,inner sep = 0pt,minimum size =5pt,fill = black] (nx15y16) at (15,16) {};
\node[circle,inner sep = 0pt,minimum size =5pt,fill = black] (nx14y17) at (14,17) {};
\node[circle,inner sep = 0pt,minimum size =0pt,fill = black] (nx15y15) at (15,15) {};
\node[circle,inner sep = 0pt,minimum size =0pt,fill = black] (nx17y17) at (17,17) {};
\node[circle,inner sep = 0pt,minimum size =0pt,fill = black] (nx15y19) at (15,19) {};
\node[circle,inner sep = 0pt,minimum size =0pt,fill = black] (nx13y17) at (13,17) {};
\node[circle,inner sep = 0pt,minimum size =0pt,fill = black] (nx7y17) at (7,17) {};
\node[circle,inner sep = 0pt,minimum size =0pt,fill = black] (nx5y15) at (5,15) {};
\node[circle,inner sep = 0pt,minimum size =0pt,fill = black] (nx3y17) at (3,17) {};
\node[circle,inner sep = 0pt,minimum size =0pt,fill = black] (nx5y19) at (5,19) {};
\node[circle,inner sep = 0pt,minimum size =5pt,fill = black] (nx5y17) at (5,17) {};
\draw[black] (nx15y16) -- (nx15y15);
\draw[black] (nx15y16) -- (nx14y17);
\draw[black] (nx16y17) -- (nx15y16);
\draw[black] (nx16y17) -- (nx17y17);
\draw[black] (nx15y18) -- (nx16y17);
\draw[black] (nx15y19) -- (nx15y18);
\draw[black] (nx15y18) -- (nx14y17);
\draw[black] (nx13y17) -- (nx14y17);
\draw[black] (nx5y17) -- (nx7y17);
\draw[black] (nx5y17) -- (nx5y15);
\draw[black] (nx5y17) -- (nx3y17);
\draw[black] (nx5y19) -- (nx5y17);
\draw (17,19) node {$G_v$};
\draw (3,19) node {$G$};
\draw (5,17) node[anchor = north west] {$v$};
\draw (10,17) node {$\Longrightarrow$};
\end{tikzpicture}
\end{center}
\legende{How to forbid the crossing of paths at vertex $v$.}
\label{fig:biflot-ext}
\end{figure}

Formally, we define the following problem, and show that it is equivalent to the original one :

\begin{problem}[Extended Planar Edge-Disjoint Paths]\hfill\\
{\sc Input}~: a planar graph $G$, a demand graph $H$ with $V(H) \subseteq V(G)$, and $U \subseteq V(G)$.\\
{\sc Output}~: Is there a solution to the edge-disjoint paths problem $(G,H)$ such that for every $u \in U$, there is no crossing at $u$ ? 
\end{problem}

Let $G_v$ be the graph obtained from $G$ by replacing a vertex $v$ of degree $4$ by a cycle of length $4$, according to Figure~\ref{fig:biflot-ext}. The equivalence with the planar edge-disjoint paths problem is a consequence of the following obvious lemma :
\begin{lemma}
Let $(G,H,U)$ be an instance of the extended planar edge-disjoint path problem. Let $v$ be a vertex of $G$ of degree $4$. Then there is a solution to $(G_v,H,U)$ if and only if there is a solution to $(G,H,U\cup \{v\})$.\qedhere
\end{lemma}

In the following, we will always study instances of the extended problem. Vertices of $U$ are called \emph{non-crossing vertices}. In the figures, we will represent vertices not in $U$ (\emph{crossing vertices}) by bold points. Note that there is no restriction over no-paths, they can cross other no-paths or paths at every vertex.

\section{Implementing $\xch$ and $\lic$}

\subsection{Basic graphs}\label{sec:elem}

We give the graphs encoding the gadgets $\xch$ and $\lic$, and detail their respective properties. Let $\xch$ be the graph depicted in Figure~\ref{fig:biflot-xch}. The crossing vertices are $a$, $b$, $c$ and $d$. We note $S = \{s_1,s_2,s_3,s_4\}$, $S' = \{s'_1,s'_2,s'_3,s'_4\}$, $T = \{t_1, t_2\}$ and $T' = \{t'_1, t'_2\}$.

\begin{figure}[htb]
\begin{center}
\begin{tikzpicture}[x=0.4cm,y=0.4cm,>=latex]
% Fichier d'origine : th-biflot-xch.pgf
\node[circle,inner sep = 0pt,minimum size =6pt,fill = black] (nx14y14) at (14,14) {};
\node[circle,inner sep = 0pt,minimum size =6pt,fill = black] (nx16y16) at (16,16) {};
\node[circle,inner sep = 0pt,minimum size =6pt,fill = black] (nx14y18) at (14,18) {};
\node[circle,inner sep = 0pt,minimum size =6pt,fill = black] (nx12y16) at (12,16) {};
\node[circle,inner sep = 0pt,minimum size = 3pt,fill = black] (nx4y10) at (4,10) {};
\node[circle,inner sep = 0pt,minimum size = 3pt,fill = black] (nx6y8) at (6,8) {};
\node[circle,inner sep = 0pt,minimum size = 3pt,fill = black] (nx10y8) at (10,8) {};
\node[circle,inner sep = 0pt,minimum size = 3pt,fill = black] (nx24y10) at (24,10) {};
\node[circle,inner sep = 0pt,minimum size = 3pt,fill = black] (nx22y8) at (22,8) {};
\node[circle,inner sep = 0pt,minimum size = 3pt,fill = black] (nx18y8) at (18,8) {};
\node[circle,inner sep = 0pt,minimum size = 3pt,fill = black] (nx22y10) at (22,10) {};
\node[circle,inner sep = 0pt,minimum size = 3pt,fill = black] (nx18y10) at (18,10) {};
\node[circle,inner sep = 0pt,minimum size = 3pt,fill = black] (nx10y10) at (10,10) {};
\node[circle,inner sep = 0pt,minimum size = 3pt,fill = black] (nx6y10) at (6,10) {};
\node[circle,inner sep = 0pt,minimum size = 3pt,fill = black] (nx24y22) at (24,22) {};
\node[circle,inner sep = 0pt,minimum size = 3pt,fill = black] (nx22y24) at (22,24) {};
\node[circle,inner sep = 0pt,minimum size = 3pt,fill = black] (nx18y24) at (18,24) {};
\node[circle,inner sep = 0pt,minimum size = 3pt,fill = black] (nx22y22) at (22,22) {};
\node[circle,inner sep = 0pt,minimum size = 3pt,fill = black] (nx18y14) at (18,14) {};
\node[circle,inner sep = 0pt,minimum size = 3pt,fill = black] (nx10y14) at (10,14) {};
\node[circle,inner sep = 0pt,minimum size = 3pt,fill = black] (nx18y18) at (18,18) {};
\node[circle,inner sep = 0pt,minimum size = 3pt,fill = black] (nx18y22) at (18,22) {};
\node[circle,inner sep = 0pt,minimum size = 3pt,fill = black] (nx10y18) at (10,18) {};
\node[circle,inner sep = 0pt,minimum size = 3pt,fill = black] (nx10y24) at (10,24) {};
\node[circle,inner sep = 0pt,minimum size = 3pt,fill = black] (nx10y22) at (10,22) {};
\node[circle,inner sep = 0pt,minimum size = 3pt,fill = black] (nx4y22) at (4,22) {};
\node[circle,inner sep = 0pt,minimum size = 3pt,fill = black] (nx6y22) at (6,22) {};
\node[circle,inner sep = 0pt,minimum size = 3pt,fill = black] (nx6y24) at (6,24) {};
\draw[black] (nx6y10) -- (nx4y10);
\draw[black] (nx6y8) -- (nx6y10);
\draw[black] (nx10y8) -- (nx10y10);
\draw[black] (nx22y10) -- (nx24y10);
\draw[black] (nx22y8) -- (nx22y10);
\draw[black] (nx18y10) -- (nx18y8);
\draw[black] (nx22y10) -- (nx18y10);
\draw[black] (nx18y14) -- (nx22y10);
\draw[black] (nx18y10) -- (nx18y14);
\draw[black] (nx14y14) -- (nx18y10);
\draw[black] (nx10y10) -- (nx14y14);
\draw[black] (nx10y10) -- (nx10y14);
\draw[black] (nx6y10) -- (nx10y10);
\draw[black] (nx10y14) -- (nx6y10);
\draw[black] (nx22y22) -- (nx24y22);
\draw[black] (nx22y24) -- (nx22y22);
\draw[black] (nx18y22) -- (nx18y24);
\draw[black] (nx22y22) -- (nx18y22);
\draw[black] (nx18y18) -- (nx22y22);
\draw[black] (nx18y14) -- (nx18y18);
\draw[black] (nx16y16) -- (nx18y14);
\draw[black] (nx14y14) -- (nx16y16);
\draw[black] (nx12y16) -- (nx14y14);
\draw[black] (nx10y14) -- (nx12y16);
\draw[black] (nx10y18) -- (nx10y14);
\draw[black] (nx12y16) -- (nx10y18);
\draw[black] (nx14y18) -- (nx12y16);
\draw[black] (nx16y16) -- (nx14y18);
\draw[black] (nx18y18) -- (nx16y16);
\draw[black] (nx18y22) -- (nx18y18);
\draw[black] (nx14y18) -- (nx18y22);
\draw[black] (nx10y22) -- (nx14y18);
\draw[black] (nx10y18) -- (nx10y22);
\draw[black] (nx6y22) -- (nx10y18);
\draw[black] (nx10y22) -- (nx10y24);
\draw[black] (nx6y22) -- (nx10y22);
\draw[black] (nx6y22) -- (nx4y22);
\draw[black] (nx6y24) -- (nx6y22);
\draw (14,14) node[anchor = south] {$d$};
\draw (16,16) node[anchor = south] {$c$};
\draw (12,16) node[anchor = south] {$b$};
\draw (14,18) node[anchor = south] {$a$};
\draw (22,10) node[anchor = north east] {$u_{12}$};
\draw (18,10) node[anchor = north east] {$u_{11}$};
\draw (10,10) node[anchor = north west] {$u_{10}$};
\draw (6,10) node[anchor = north west] {$u_9$};
\draw (18,14) node[anchor = south west] {$u_8$};
\draw (10,14) node[anchor = south east] {$u_7$};
\draw (18,18) node[anchor = north west] {$u_6$};
\draw (10,18) node[anchor = north east] {$u_5$};
\draw (22,22) node[anchor = south east] {$u_4$};
\draw (18,22) node[anchor = south east] {$u_3$};
\draw (10,22) node[anchor = south west] {$u_2$};
\draw (6,22) node[anchor = south west] {$u_1$};
\draw (24,10) node[anchor = west] {$t'_2$};
\draw (24,22) node[anchor = west] {$t'_1$};
\draw (4,10) node[anchor = east] {$t_2$};
\draw (4,22) node[anchor = east] {$t_1$};
\draw (22,8) node[anchor = north] {$s'_4$};
\draw (18,8) node[anchor = north] {$s'_3$};
\draw (10,8) node[anchor = north] {$s'_2$};
\draw (6,8) node[anchor = north] {$s'_1$};
\draw (22,24) node[anchor = south] {$s_4$};
\draw (18,24) node[anchor = south] {$s_3$};
\draw (10,24) node[anchor = south] {$s_2$};
\draw (6,24) node[anchor = south] {$s_1$};
\end{tikzpicture}
\end{center}
\legende{The graph $\xch$. There are only four crossing vertices $a$, $b$, $c$ and $d$. All the \emph{inner} vertices have degree $4$.}
\label{fig:biflot-xch}
\end{figure}

\begin{lemma}\label{lemme:xch-22}
Let $\mathcal{P} = \{S_1,S_2,T_1,T_2\}$ be an uncrossed set of edge-disjoint paths in $\xch$, satisfying :
\begin{enumerate}
\item[$(i)$] $S_1$ and $S_2$ are $(S,S')$-paths,
\item[$(ii)$] $T_1$ and $T_2$ are $(T,T')$-paths.
\end{enumerate}
Then $S_1$ and $S_2$ are either an $(s_1,s'_3)$-path and an $(s_2,s'_4)$-path, or an $(s_3,s'_1)$-path and an $(s_4,s'_2)$-path.
\end{lemma}

\begin{proof}
Let $\mathcal{P}$ be as described in the lemma. As the $(S,S')$-paths must cross the $(T,T')$-paths, there are at least $4$ crossings in $\mathcal{P}$. We know that these crossings occur in vertices $a$, $b$, $c$ and $d$, and the $(T,T')$-paths (resp. the $(S,S')$-path) do not cross each other.\\
Suppose $ab$ is in an $(S,S')$-path, then $ac$ and $bd$ must belong to distinct $(T,T')$-paths, and $cd$ is in the second $(S,S')$-path. Then $au_2$, $bu_5$, $cu_8$, $du_{11}$ are in $(T,T')$ paths, and $au_3$, $bu_7$, $cu_6$ and $du_{10}$ are in $(S,S')$-paths. As there is no other crossing except in the four central vertices, the $(S,S')$-paths are connected to $s'_1$, $s'_2$, $s_3$ and $s_4$. The case when $ab$ belongs to a $(T,T')$-path is similar and gives the other solution.
\end{proof}

These paths exist, as shown by Figure~\ref{fig:biflot-xch-22}.

\begin{figure}[htb]
\begin{center}
\begin{tikzpicture}[x=0.4cm,y=0.4cm,>=latex]
% Fichier d'origine : th-biflot-xch-22.pgf
\definecolor{lightg}{rgb}{0.8,0.8,0.8}
\node[circle,inner sep = 0pt,minimum size =6pt,fill = black] (nx6y19) at (6,19) {};
\node[circle,inner sep = 0pt,minimum size =6pt,fill = black] (nx7y20) at (7,20) {};
\node[circle,inner sep = 0pt,minimum size =6pt,fill = black] (nx6y21) at (6,21) {};
\node[circle,inner sep = 0pt,minimum size =6pt,fill = black] (nx5y20) at (5,20) {};
\node[circle,inner sep = 0pt,minimum size = 3pt,fill = black] (nx12y17) at (12,17) {};
\node[circle,inner sep = 0pt,minimum size = 3pt,fill = black] (nx8y19) at (8,19) {};
\node[circle,inner sep = 0pt,minimum size = 3pt,fill = black] (nx10y17) at (10,17) {};
\node[circle,inner sep = 0pt,minimum size = 3pt,fill = black] (nx10y15) at (10,15) {};
\node[circle,inner sep = 0pt,minimum size = 3pt,fill = black] (nx8y15) at (8,15) {};
\node[circle,inner sep = 0pt,minimum size = 3pt,fill = black] (nx8y17) at (8,17) {};
\node[circle,inner sep = 0pt,minimum size = 3pt,fill = black] (nx0y17) at (0,17) {};
\node[circle,inner sep = 0pt,minimum size = 3pt,fill = black] (nx2y15) at (2,15) {};
\node[circle,inner sep = 0pt,minimum size = 3pt,fill = black] (nx2y17) at (2,17) {};
\node[circle,inner sep = 0pt,minimum size = 3pt,fill = black] (nx4y17) at (4,17) {};
\node[circle,inner sep = 0pt,minimum size = 3pt,fill = black] (nx4y15) at (4,15) {};
\node[circle,inner sep = 0pt,minimum size = 3pt,fill = black] (nx12y23) at (12,23) {};
\node[circle,inner sep = 0pt,minimum size = 3pt,fill = black] (nx8y21) at (8,21) {};
\node[circle,inner sep = 0pt,minimum size = 3pt,fill = black] (nx10y23) at (10,23) {};
\node[circle,inner sep = 0pt,minimum size = 3pt,fill = black] (nx10y25) at (10,25) {};
\node[circle,inner sep = 0pt,minimum size = 3pt,fill = black] (nx8y25) at (8,25) {};
\node[circle,inner sep = 0pt,minimum size = 3pt,fill = black] (nx8y23) at (8,23) {};
\node[circle,inner sep = 0pt,minimum size = 3pt,fill = black] (nx4y19) at (4,19) {};
\node[circle,inner sep = 0pt,minimum size = 3pt,fill = black] (nx4y25) at (4,25) {};
\node[circle,inner sep = 0pt,minimum size = 3pt,fill = black] (nx0y23) at (0,23) {};
\node[circle,inner sep = 0pt,minimum size = 3pt,fill = black] (nx2y25) at (2,25) {};
\node[circle,inner sep = 0pt,minimum size = 3pt,fill = black] (nx4y21) at (4,21) {};
\node[circle,inner sep = 0pt,minimum size = 3pt,fill = black] (nx4y23) at (4,23) {};
\node[circle,inner sep = 0pt,minimum size = 3pt,fill = black] (nx2y23) at (2,23) {};
\draw[black,dashed,thick] (nx10y17) -- (nx12y17);
\draw[black,dashed,thick] (nx8y19) -- (nx10y17);
\draw[black,dashed,thick] (nx8y21) -- (nx8y19);
\draw[black,dashed,thick] (nx7y20) -- (nx8y21);
\draw[black,dotted,thick] (nx8y19) -- (nx7y20);
\draw[black,dotted,thick] (nx8y17) -- (nx8y19);
\draw[black,dotted,thick] (nx10y17) -- (nx8y17);
\draw[black,dotted,thick] (nx10y15) -- (nx10y17);
\draw[black,dotted,thick] (nx8y17) -- (nx8y15);
\draw[black,dotted,thick] (nx6y19) -- (nx8y17);
\draw[black,dashed,thick] (nx4y17) -- (nx6y19);
\draw[lightg] (nx4y19) -- (nx4y17);
\draw[lightg] (nx2y17) -- (nx4y19);
\draw[black,dashed,thick] (nx0y17) -- (nx2y17);
\draw[lightg] (nx2y17) -- (nx2y15);
\draw[black,dashed,thick] (nx4y17) -- (nx2y17);
\draw[lightg] (nx4y15) -- (nx4y17);
\draw[black,dashed,thick] (nx10y23) -- (nx12y23);
\draw[lightg] (nx8y21) -- (nx10y23);
\draw[lightg] (nx8y23) -- (nx8y21);
\draw[black,dashed,thick] (nx10y23) -- (nx8y23);
\draw[lightg] (nx10y25) -- (nx10y23);
\draw[lightg] (nx8y23) -- (nx8y25);
\draw[black,dashed,thick] (nx6y21) -- (nx8y23);
\draw[black,dotted,thick] (nx6y19) -- (nx5y20);
\draw[black,dashed,thick] (nx7y20) -- (nx6y19);
\draw[black,dotted,thick] (nx6y21) -- (nx7y20);
\draw[black,dashed,thick] (nx5y20) -- (nx6y21);
\draw[black,dashed,thick] (nx4y19) -- (nx5y20);
\draw[black,dashed,thick] (nx4y21) -- (nx4y19);
\draw[black,dotted,thick] (nx4y23) -- (nx4y21);
\draw[black,dotted,thick] (nx4y25) -- (nx4y23);
\draw[black,dashed,thick] (nx2y23) -- (nx0y23);
\draw[black,dotted,thick] (nx2y23) -- (nx2y25);
\draw[black,dashed,thick] (nx4y21) -- (nx2y23);
\draw[black,dotted,thick] (nx5y20) -- (nx4y21);
\draw[black,dotted,thick] (nx4y23) -- (nx6y21);
\draw[black,dotted,thick] (nx2y23) -- (nx4y23);
\end{tikzpicture}
\end{center}
\legende{Existence of the paths for Lemma~\ref{lemme:xch-22}.}
\label{fig:biflot-xch-22}
\end{figure}

\begin{lemma}\label{lemme:xch-12}
Let $\mathcal{P}$ be a set of three edge-disjoint paths in $\xch$ that satisfies~:
\begin{enumerate}
\item[$(i)$] $\mathcal{P}$ contains exactly two $(\{s_1,s_2\},S'')$-paths, and $S''$ is either $\{s'_1,s'_2\}$ or $\{s'_3,s'_4\}$,
\item[$(ii)$] $\mathcal{P}$ contains exactly one $(T,T')$-path.
\end{enumerate}
Then, $S'' = \{s'_3,s'_4\}$.
\end{lemma}

Here we find the first differences between the graph $\xch$ and its ideal model given in Section~\ref{sec:outline}. We need to suppose that the two vertical paths come from the same side (through $s_1$ and $s_2$ or $s_3$ and $s_4$), and leave also at the same side . We need to prove that this is actually the case, but we can already remark that it is true in the particular case when the gadgets above and below are $\xch$ with two horizontal paths, by Lemma~\ref{lemme:xch-22}. 

\begin{proof}
If not, then there is a set $\mathcal{P}$ of three edge-disjoint paths, a $(t,t')$-path $Q$, an $(s_1,s'_1)$-path $P_1$ and an $(s_2,s'_2)$-path $P_2$. As $Q$ must cross the two other paths, all paths contain at least one of $a,b,c$ and $d$. Then, $Q$ uses one edge of $u_2a$, $u_5b$, $u_7b$, $u_{10}d$, and both $P_1$ and $P_2$ use two of these edges, contradicting the edge-disjointness of the paths.
\end{proof}

Let $\lic$ be the graph depicted in Figure~\ref{fig:biflot:lic}. We note again $S = \{s_1,s_2,s_3,s_4\}$, $S' = \{s'_1,s'_2,s'_3,s'_4\}$, $T = \{t_1, t_2\}$ and $T' = \{t'_1, t'_2\}$.

\begin{figure}
\begin{center}
\begin{tikzpicture}[x=0.4cm,y=0.4cm,>=latex]
% Fichier d'origine : th-biflot-xch.pgf
\node[circle,inner sep = 0pt,minimum size =6pt,fill = black] (nx14y14) at (14,14) {};
\node[circle,inner sep = 0pt,minimum size =6pt,fill = black] (nx16y16) at (16,16) {};
\node[circle,inner sep = 0pt,minimum size =6pt,fill = black] (nx14y18) at (14,18) {};
\node[circle,inner sep = 0pt,minimum size =6pt,fill = black] (nx12y16) at (12,16) {};
\node[circle,inner sep = 0pt,minimum size = 3pt,fill = black] (nx4y10) at (4,10) {};
\node[circle,inner sep = 0pt,minimum size = 3pt,fill = black] (nx6y8) at (6,8) {};
\node[circle,inner sep = 0pt,minimum size = 3pt,fill = black] (nx10y8) at (10,8) {};
\node[circle,inner sep = 0pt,minimum size = 3pt,fill = black] (nx24y10) at (24,10) {};
\node[circle,inner sep = 0pt,minimum size = 3pt,fill = black] (nx22y8) at (22,8) {};
\node[circle,inner sep = 0pt,minimum size = 3pt,fill = black] (nx18y8) at (18,8) {};
\node[circle,inner sep = 0pt,minimum size = 3pt,fill = black] (nx22y10) at (22,10) {};
\node[circle,inner sep = 0pt,minimum size = 3pt,fill = black] (nx18y10) at (18,10) {};
\node[circle,inner sep = 0pt,minimum size = 3pt,fill = black] (nx10y10) at (10,10) {};
\node[circle,inner sep = 0pt,minimum size = 3pt,fill = black] (nx6y10) at (6,10) {};
\node[circle,inner sep = 0pt,minimum size = 3pt,fill = black] (nx24y22) at (24,22) {};
\node[circle,inner sep = 0pt,minimum size = 3pt,fill = black] (nx22y24) at (22,24) {};
\node[circle,inner sep = 0pt,minimum size = 3pt,fill = black] (nx18y24) at (18,24) {};
\node[circle,inner sep = 0pt,minimum size = 3pt,fill = black] (nx22y22) at (22,22) {};
\node[circle,inner sep = 0pt,minimum size = 3pt,fill = black] (nx18y14) at (18,14) {};
\node[circle,inner sep = 0pt,minimum size = 3pt,fill = black] (nx10y14) at (10,14) {};
\node[circle,inner sep = 0pt,minimum size = 6pt,fill = black] (nx18y18) at (18,18) {};
\node[circle,inner sep = 0pt,minimum size = 3pt,fill = black] (nx18y22) at (18,22) {};
\node[circle,inner sep = 0pt,minimum size = 6pt,fill = black] (nx10y18) at (10,18) {};
\node[circle,inner sep = 0pt,minimum size = 3pt,fill = black] (nx10y24) at (10,24) {};
\node[circle,inner sep = 0pt,minimum size = 3pt,fill = black] (nx10y22) at (10,22) {};
\node[circle,inner sep = 0pt,minimum size = 3pt,fill = black] (nx4y22) at (4,22) {};
\node[circle,inner sep = 0pt,minimum size = 3pt,fill = black] (nx6y22) at (6,22) {};
\node[circle,inner sep = 0pt,minimum size = 3pt,fill = black] (nx6y24) at (6,24) {};
\draw[black] (nx6y10) -- (nx4y10);
\draw[black] (nx6y8) -- (nx6y10);
\draw[black] (nx10y8) -- (nx10y10);
\draw[black] (nx22y10) -- (nx24y10);
\draw[black] (nx22y8) -- (nx22y10);
\draw[black] (nx18y10) -- (nx18y8);
\draw[black] (nx22y10) -- (nx18y10);
\draw[black] (nx18y14) -- (nx22y10);
\draw[black] (nx18y10) -- (nx18y14);
\draw[black] (nx14y14) -- (nx18y10);
\draw[black] (nx10y10) -- (nx14y14);
\draw[black] (nx10y10) -- (nx10y14);
\draw[black] (nx6y10) -- (nx10y10);
\draw[black] (nx10y14) -- (nx6y10);
\draw[black] (nx22y22) -- (nx24y22);
\draw[black] (nx22y24) -- (nx22y22);
\draw[black] (nx18y22) -- (nx18y24);
\draw[black] (nx22y22) -- (nx18y22);
\draw[black] (nx18y18) -- (nx22y22);
\draw[black] (nx18y14) -- (nx18y18);
\draw[black] (nx16y16) -- (nx18y14);
\draw[black] (nx14y14) -- (nx16y16);
\draw[black] (nx12y16) -- (nx14y14);
\draw[black] (nx10y14) -- (nx12y16);
\draw[black] (nx10y18) -- (nx10y14);
\draw[black] (nx12y16) -- (nx10y18);
\draw[black] (nx14y18) -- (nx12y16);
\draw[black] (nx16y16) -- (nx14y18);
\draw[black] (nx18y18) -- (nx16y16);
\draw[black] (nx18y22) -- (nx18y18);
\draw[black] (nx14y18) -- (nx18y22);
\draw[black] (nx10y22) -- (nx14y18);
\draw[black] (nx10y18) -- (nx10y22);
\draw[black] (nx6y22) -- (nx10y18);
\draw[black] (nx10y22) -- (nx10y24);
\draw[black] (nx6y22) -- (nx10y22);
\draw[black] (nx6y22) -- (nx4y22);
\draw[black] (nx6y24) -- (nx6y22);
\draw (14,14) node[anchor = south] {$d$};
\draw (16,16) node[anchor = south] {$c$};
\draw (12,16) node[anchor = south] {$b$};
\draw (14,18) node[anchor = south] {$a$};
\draw (22,10) node[anchor = north east] {$u_{12}$};
\draw (18,10) node[anchor = north east] {$u_{11}$};
\draw (10,10) node[anchor = north west] {$u_{10}$};
\draw (6,10) node[anchor = north west] {$u_9$};
\draw (18,14) node[anchor = south west] {$u_8$};
\draw (10,14) node[anchor = south east] {$u_7$};
\draw (18,18) node[anchor = north west] {$u_6$};
\draw (10,18) node[anchor = north east] {$u_5$};
\draw (22,22) node[anchor = south east] {$u_4$};
\draw (18,22) node[anchor = south east] {$u_3$};
\draw (10,22) node[anchor = south west] {$u_2$};
\draw (6,22) node[anchor = south west] {$u_1$};
\draw (24,10) node[anchor = west] {$t'_2$};
\draw (24,22) node[anchor = west] {$t'_1$};
\draw (4,10) node[anchor = east] {$t_2$};
\draw (4,22) node[anchor = east] {$t_1$};
\draw (22,8) node[anchor = north] {$s'_4$};
\draw (18,8) node[anchor = north] {$s'_3$};
\draw (10,8) node[anchor = north] {$s'_2$};
\draw (6,8) node[anchor = north] {$s'_1$};
\draw (22,24) node[anchor = south] {$s_4$};
\draw (18,24) node[anchor = south] {$s_3$};
\draw (10,24) node[anchor = south] {$s_2$};
\draw (6,24) node[anchor = south] {$s_1$};
\end{tikzpicture}
\end{center}
\legende{Graph $\lic$. Except two additional crossing vertices, it is similar to $\xch$.}
\label{fig:biflot:lic}
\end{figure}

\begin{lemma}\label{lemme:lic-22}
Let $\mathcal{P}$ be a set of edge-disjoint paths in $\textrm{LIC}$ that satisfy :
\begin{itemize}
\item[$(i)$] $\mathcal{P}$ contains exactly $2$ $(\{s_1,s_2\},S'')$-paths, where $S''$ is either $\{s'_1,s'_2\}$ or $\{s'_3,s'_4\}$,
\item[$(ii)$] $\mathcal{P}$ contains exactly $2$ $(T,T')$-paths.
\end{itemize}
Then, $S'' = \{s'_3,s'_4\}$. Moreover, there cannot be another $(S \cup T, S' \cup T')$-path.
\end{lemma}

\begin{proof}
Suppose $S'' = \{s'_1,s'_2\}$. Let $C$ be the cut $\{u_7,u_9,u_{10}\}$. Let $Q$ be the $(t_2,T')$-path, $P_1$ the $(s_1,s'_1)$-path and $P_2$ the $(s_2,s'_2)$-path. All three different paths meet $C$, and $d(C) = 6$. Moreover, there is no crossing in $C$, thus considering $\delta(C)$, $u_5u_7$ is used by $Q$ and $bu_7$, $du_{10}$ by $P_1$ and $P_2$. Now, there are four distinct paths entering $C' = \{u_1,u_2,u_5\}$, but $d(C') = 6$, contradiction.\\
As the edges of $\delta(\{a,b,c,d\})$ are all used by $\mathcal{P}$, there is no other $(S \cup T, S' \cup T')$-path.
\end{proof}

\begin{lemma}\label{lemme:lic-12}
There exist a $(T,T')$-path $P$, an $(s_1,s'_1)$-path $P_1$ and an $(s_2,s'_2)$-path $P_2$, pairwise edge-disjoint, in $\textrm{LIC}$.\\
There exist a $(T,T')$-path $P$, an $(s_3,s'_3)$-path $P_1$ and an $(s_4,s'_4)$-path $P_2$, pairwise edge-disjoint, in $\textrm{LIC}$.\\
\end{lemma}

\begin{proof}
See Figure~\ref{fig:lic-12}.
\end{proof}

\begin{figure}
\begin{center}
\begin{tikzpicture}[x=0.4cm,y=0.4cm,>=latex]
% Fichier d'origine : th-biflot-lic-12.pgf
\definecolor{lightg}{rgb}{0.8,0.8,0.8}
\node[circle,inner sep = 0pt,minimum size =6pt,fill = black] (nx24y21) at (24,21) {};
\node[circle,inner sep = 0pt,minimum size =6pt,fill = black] (nx23y20) at (23,20) {};
\node[circle,inner sep = 0pt,minimum size =6pt,fill = black] (nx22y19) at (22,19) {};
\node[circle,inner sep = 0pt,minimum size =6pt,fill = black] (nx21y20) at (21,20) {};
\node[circle,inner sep = 0pt,minimum size =6pt,fill = black] (nx22y21) at (22,21) {};
\node[circle,inner sep = 0pt,minimum size =6pt,fill = black] (nx20y21) at (20,21) {};
\node[circle,inner sep = 0pt,minimum size = 3pt,fill = black] (nx16y17) at (16,17) {};
\node[circle,inner sep = 0pt,minimum size = 3pt,fill = black] (nx18y15) at (18,15) {};
\node[circle,inner sep = 0pt,minimum size = 3pt,fill = black] (nx20y15) at (20,15) {};
\node[circle,inner sep = 0pt,minimum size = 3pt,fill = black] (nx18y17) at (18,17) {};
\node[circle,inner sep = 0pt,minimum size = 3pt,fill = black] (nx20y17) at (20,17) {};
\node[circle,inner sep = 0pt,minimum size = 3pt,fill = black] (nx24y15) at (24,15) {};
\node[circle,inner sep = 0pt,minimum size = 3pt,fill = black] (nx26y15) at (26,15) {};
\node[circle,inner sep = 0pt,minimum size = 3pt,fill = black] (nx24y17) at (24,17) {};
\node[circle,inner sep = 0pt,minimum size = 3pt,fill = black] (nx26y17) at (26,17) {};
\node[circle,inner sep = 0pt,minimum size = 3pt,fill = black] (nx28y17) at (28,17) {};
\node[circle,inner sep = 0pt,minimum size = 3pt,fill = black] (nx28y23) at (28,23) {};
\node[circle,inner sep = 0pt,minimum size = 3pt,fill = black] (nx26y25) at (26,25) {};
\node[circle,inner sep = 0pt,minimum size = 3pt,fill = black] (nx24y25) at (24,25) {};
\node[circle,inner sep = 0pt,minimum size = 3pt,fill = black] (nx24y19) at (24,19) {};
\node[circle,inner sep = 0pt,minimum size = 3pt,fill = black] (nx20y19) at (20,19) {};
\node[circle,inner sep = 0pt,minimum size = 3pt,fill = black] (nx18y25) at (18,25) {};
\node[circle,inner sep = 0pt,minimum size = 3pt,fill = black] (nx20y25) at (20,25) {};
\node[circle,inner sep = 0pt,minimum size = 3pt,fill = black] (nx26y23) at (26,23) {};
\node[circle,inner sep = 0pt,minimum size = 3pt,fill = black] (nx24y23) at (24,23) {};
\node[circle,inner sep = 0pt,minimum size = 3pt,fill = black] (nx20y23) at (20,23) {};
\node[circle,inner sep = 0pt,minimum size = 3pt,fill = black] (nx18y23) at (18,23) {};
\node[circle,inner sep = 0pt,minimum size = 3pt,fill = black] (nx16y23) at (16,23) {};
\node[circle,inner sep = 0pt,minimum size =6pt,fill = black] (nx7y19) at (7,19) {};
\node[circle,inner sep = 0pt,minimum size =6pt,fill = black] (nx6y20) at (6,20) {};
\node[circle,inner sep = 0pt,minimum size =6pt,fill = black] (nx8y20) at (8,20) {};
\node[circle,inner sep = 0pt,minimum size =6pt,fill = black] (nx9y21) at (9,21) {};
\node[circle,inner sep = 0pt,minimum size =6pt,fill = black] (nx7y21) at (7,21) {};
\node[circle,inner sep = 0pt,minimum size =6pt,fill = black] (nx5y21) at (5,21) {};
\node[circle,inner sep = 0pt,minimum size = 3pt,fill = black] (nx1y17) at (1,17) {};
\node[circle,inner sep = 0pt,minimum size = 3pt,fill = black] (nx3y15) at (3,15) {};
\node[circle,inner sep = 0pt,minimum size = 3pt,fill = black] (nx5y15) at (5,15) {};
\node[circle,inner sep = 0pt,minimum size = 3pt,fill = black] (nx9y15) at (9,15) {};
\node[circle,inner sep = 0pt,minimum size = 3pt,fill = black] (nx11y15) at (11,15) {};
\node[circle,inner sep = 0pt,minimum size = 3pt,fill = black] (nx13y17) at (13,17) {};
\node[circle,inner sep = 0pt,minimum size = 3pt,fill = black] (nx13y23) at (13,23) {};
\node[circle,inner sep = 0pt,minimum size = 3pt,fill = black] (nx11y25) at (11,25) {};
\node[circle,inner sep = 0pt,minimum size = 3pt,fill = black] (nx9y25) at (9,25) {};
\node[circle,inner sep = 0pt,minimum size = 3pt,fill = black] (nx9y17) at (9,17) {};
\node[circle,inner sep = 0pt,minimum size = 3pt,fill = black] (nx11y17) at (11,17) {};
\node[circle,inner sep = 0pt,minimum size = 3pt,fill = black] (nx9y19) at (9,19) {};
\node[circle,inner sep = 0pt,minimum size = 3pt,fill = black] (nx11y23) at (11,23) {};
\node[circle,inner sep = 0pt,minimum size = 3pt,fill = black] (nx9y23) at (9,23) {};
\node[circle,inner sep = 0pt,minimum size = 3pt,fill = black] (nx5y17) at (5,17) {};
\node[circle,inner sep = 0pt,minimum size = 3pt,fill = black] (nx3y17) at (3,17) {};
\node[circle,inner sep = 0pt,minimum size = 3pt,fill = black] (nx5y19) at (5,19) {};
\node[circle,inner sep = 0pt,minimum size = 3pt,fill = black] (nx3y25) at (3,25) {};
\node[circle,inner sep = 0pt,minimum size = 3pt,fill = black] (nx5y25) at (5,25) {};
\node[circle,inner sep = 0pt,minimum size = 3pt,fill = black] (nx1y23) at (1,23) {};
\node[circle,inner sep = 0pt,minimum size = 3pt,fill = black] (nx5y23) at (5,23) {};
\node[circle,inner sep = 0pt,minimum size = 3pt,fill = black] (nx3y23) at (3,23) {};
\draw[lightg] (nx18y17) -- (nx16y17);
\draw[black,dotted,thick] (nx18y15) -- (nx18y17);
\draw[black,dotted,thick] (nx20y17) -- (nx20y15);
\draw[lightg] (nx18y17) -- (nx20y17);
\draw[black,dotted,thick] (nx20y19) -- (nx18y17);
\draw[lightg] (nx20y17) -- (nx20y19);
\draw[black,dotted,thick] (nx22y19) -- (nx20y17);
\draw[black,dashed,thick] (nx24y17) -- (nx22y19);
\draw[lightg] (nx24y15) -- (nx24y17);
\draw[lightg] (nx26y17) -- (nx26y15);
\draw[lightg] (nx24y19) -- (nx26y17);
\draw[lightg] (nx24y17) -- (nx24y19);
\draw[black,dashed,thick] (nx26y17) -- (nx24y17);
\draw[black,dashed,thick] (nx28y17) -- (nx26y17);
\draw[lightg] (nx26y23) -- (nx28y23);
\draw[lightg] (nx26y25) -- (nx26y23);
\draw[lightg] (nx24y23) -- (nx24y25);
\draw[lightg] (nx24y21) -- (nx24y23);
\draw[lightg] (nx24y19) -- (nx24y21);
\draw[lightg] (nx23y20) -- (nx24y19);
\draw[black,dotted,thick] (nx22y19) -- (nx23y20);
\draw[black,dashed,thick] (nx21y20) -- (nx22y19);
\draw[lightg] (nx20y19) -- (nx21y20);
\draw[black,dotted,thick] (nx20y21) -- (nx20y19);
\draw[black,dashed,thick] (nx18y23) -- (nx20y21);
\draw[black,dotted,thick] (nx18y25) -- (nx18y23);
\draw[black,dotted,thick] (nx20y23) -- (nx20y25);
\draw[black,dotted,thick] (nx20y21) -- (nx20y23);
\draw[black,dashed,thick] (nx21y20) -- (nx20y21);
\draw[lightg] (nx22y21) -- (nx21y20);
\draw[black,dotted,thick] (nx23y20) -- (nx22y21);
\draw[lightg] (nx24y21) -- (nx23y20);
\draw[lightg] (nx26y23) -- (nx24y21);
\draw[lightg] (nx24y23) -- (nx26y23);
\draw[lightg] (nx22y21) -- (nx24y23);
\draw[black,dotted,thick] (nx20y23) -- (nx22y21);
\draw[black,dotted,thick] (nx18y23) -- (nx20y23);
\draw[black,dashed,thick] (nx16y23) -- (nx18y23);
\draw[lightg] (nx3y17) -- (nx1y17);
\draw[black,dotted,thick] (nx3y15) -- (nx3y17);
\draw[black,dotted,thick] (nx5y17) -- (nx5y15);
\draw[black,dotted,thick] (nx7y19) -- (nx5y17);
\draw[lightg] (nx9y17) -- (nx7y19);
\draw[lightg] (nx9y15) -- (nx9y17);
\draw[lightg] (nx11y17) -- (nx11y15);
\draw[lightg] (nx13y17) -- (nx11y17);
\draw[black,dashed,thick] (nx11y23) -- (nx13y23);
\draw[lightg] (nx11y25) -- (nx11y23);
\draw[lightg] (nx9y23) -- (nx9y25);
\draw[lightg] (nx9y21) -- (nx9y23);
\draw[lightg] (nx9y19) -- (nx9y21);
\draw[lightg] (nx9y17) -- (nx9y19);
\draw[lightg] (nx11y17) -- (nx9y17);
\draw[lightg] (nx9y19) -- (nx11y17);
\draw[lightg] (nx8y20) -- (nx9y19);
\draw[lightg] (nx9y21) -- (nx8y20);
\draw[lightg] (nx11y23) -- (nx9y21);
\draw[black,dashed,thick] (nx9y23) -- (nx11y23);
\draw[black,dashed,thick] (nx7y21) -- (nx9y23);
\draw[black,dashed,thick] (nx6y20) -- (nx7y21);
\draw[lightg] (nx7y19) -- (nx6y20);
\draw[black,dotted,thick] (nx8y20) -- (nx7y19);
\draw[black,dotted,thick] (nx7y21) -- (nx8y20);
\draw[black,dotted,thick] (nx5y23) -- (nx7y21);
\draw[black,dashed,thick] (nx6y20) -- (nx5y21);
\draw[lightg] (nx5y19) -- (nx6y20);
\draw[lightg] (nx5y17) -- (nx5y19);
\draw[lightg] (nx3y17) -- (nx5y17);
\draw[black,dotted,thick] (nx5y19) -- (nx3y17);
\draw[black,dotted,thick] (nx5y21) -- (nx5y19);
\draw[black,dotted,thick] (nx3y25) -- (nx3y23);
\draw[black,dotted,thick] (nx5y23) -- (nx5y25);
\draw[black,dotted,thick] (nx5y21) -- (nx5y23);
\draw[black,dashed,thick] (nx3y23) -- (nx5y21);
\draw[black,dashed,thick] (nx3y23) -- (nx1y23);
\draw[black,dotted,thick] (nx3y23) -- (nx5y23);
\end{tikzpicture}
\end{center}
\legende{Existence of the paths for Lemma~\ref{lemme:lic-12}.}
\label{fig:lic-12}
\end{figure}

With these four lemmas, the ideal behaviour of $\xch$ and $\lic$ is not reached by the two gadgets: paths can still go from $S \cup S'$ to $T \cup T'$, and we did not prove that the vertical paths must use either the two rightmost edges or the two leftmost edges between two consecutives gadgets.

\subsection{Aggregating gadgets}

In order to build the graph for the reduction, we need to aggregate gadgets in the form of a grid, using $\xch$ and $\lic$ in the crossing of rows and columns. This is done by linking the edges incident to corresponding vertices of degree $1$. Figure~\ref{fig:biflot-not} shows how we build the grid, and gives the notation that we will use. Moreover, we define $X := \{x_i~:~i \in \ival{1}{4n}\}$ and $X' := \{x'_i~:~i \in \ival{1}{4n}\}$ where $n$ is the number of columns, and similarly $Y := \{y_i~:~i \in \ival{1}{2p}\}$ and $Y' := \{y'_i~:~i \in \ival{1}{2p}\}$ where $p$ is the number of rows. We call $i^{\textit{th}}$ \emph{vertical cut} the set of edges $V_i := \{f_k^{i,j}~:~j \in \ival{1}{p}, k \in \ival{1}{2}\}$, and $j^{\textit{th}}$ \emph{horizontal cut} the set of edges $H_j := \{e_k^{i,j}~:~i \in \ival{1}{n}, k \in \ival{1}{4}\}$. Vertex $v$ in $M(i,j)$ will be denoted $v^{i,j}$.\\

\begin{figure}[htbp]
\thispagestyle{plain}
\begin{center}

\end{center}
\legende{A grid of dimension $2 \times 3$.}
\label{fig:biflot-not}
\end{figure}

The following lemma will help us to explain why the horizontal paths cannot go through more than two different gadgets in each column, justifying the buffers. More exactly, the no-paths can be vertically moved by at most $2$ rows in each column, so the buffers must contain at least $4$ times the number of columns plus $2$.

\begin{lemma}\label{lemme:decalage-nc}
Let $G$ be a grid of dimension $1 \times 3$, built exclusively with $\lic$. Let $\mathcal{P}$ be a set of edge-disjoint paths, consisting of two $(X,X')$-paths $A$ and $B$, a $(\{y_1,y_2\},y'_1)$-path $C$, a $(y_3,y'_2)$-path $D$, a $(y_4,y'_3)$-path $E$, a $(y_5,y'_4)$-path $F$, a $(y_6,y'_5)$-path $H$ and an $(X',y'_6)$-path $I$. Then $y_1$ and $y_2$ are disconnected from $X'$ in $G \setminus E(\mathcal{P})$ (or equivalently, there is no no-path from $\{y_1,y_2\}$ to $X'$).
\end{lemma}

Note that this lemma is still true if some of the $\lic$ graphs are replaced by $\xch$. Figure~\ref{fig:biflot-exclaim5} shows the extremities of each path.

\begin{figure}[htb]
\begin{center}
\begin{tikzpicture}[x=0.3cm,y=0.3cm,>=latex]
% Fichier d'origine : brouillon.pgf
\node[circle,inner sep = 0pt,minimum size = 3pt,fill = black] (nx13y28) at (13,28) {};
\node[circle,inner sep = 0pt,minimum size = 3pt,fill = black] (nx13y0) at (13,0) {};
\node[circle,inner sep = 0pt,minimum size = 3pt,fill = black] (nx17y2) at (17,2) {};
\node[circle,inner sep = 0pt,minimum size = 3pt,fill = black] (nx15y2) at (15,2) {};
\node[circle,inner sep = 0pt,minimum size = 3pt,fill = black] (nx11y2) at (11,2) {};
\node[circle,inner sep = 0pt,minimum size = 3pt,fill = black] (nx9y2) at (9,2) {};
\node[circle,inner sep = 0pt,minimum size = 3pt,fill = black] (nx17y26) at (17,26) {};
\node[circle,inner sep = 0pt,minimum size = 3pt,fill = black] (nx15y26) at (15,26) {};
\node[circle,inner sep = 0pt,minimum size = 3pt,fill = black] (nx11y26) at (11,26) {};
\node[circle,inner sep = 0pt,minimum size = 3pt,fill = black] (nx9y26) at (9,26) {};
\node[circle,inner sep = 0pt,minimum size = 3pt,fill = black] (nx5y22) at (5,22) {};
\node[circle,inner sep = 0pt,minimum size = 3pt,fill = black] (nx7y9) at (7,9) {};
\node[circle,inner sep = 0pt,minimum size = 3pt,fill = black] (nx19y9) at (19,9) {};
\node[circle,inner sep = 0pt,minimum size =6pt,fill = black] (nx13y5) at (13,5) {};
\node[circle,inner sep = 0pt,minimum size =6pt,fill = black] (nx12y6) at (12,6) {};
\node[circle,inner sep = 0pt,minimum size =6pt,fill = black] (nx14y6) at (14,6) {};
\node[circle,inner sep = 0pt,minimum size =6pt,fill = black] (nx15y7) at (15,7) {};
\node[circle,inner sep = 0pt,minimum size =6pt,fill = black] (nx13y7) at (13,7) {};
\node[circle,inner sep = 0pt,minimum size =6pt,fill = black] (nx11y7) at (11,7) {};
\node[circle,inner sep = 0pt,minimum size =6pt,fill = black] (nx13y13) at (13,13) {};
\node[circle,inner sep = 0pt,minimum size =6pt,fill = black] (nx12y14) at (12,14) {};
\node[circle,inner sep = 0pt,minimum size =6pt,fill = black] (nx14y14) at (14,14) {};
\node[circle,inner sep = 0pt,minimum size =6pt,fill = black] (nx15y15) at (15,15) {};
\node[circle,inner sep = 0pt,minimum size =6pt,fill = black] (nx13y15) at (13,15) {};
\node[circle,inner sep = 0pt,minimum size =6pt,fill = black] (nx11y15) at (11,15) {};
\node[circle,inner sep = 0pt,minimum size =6pt,fill = black] (nx13y21) at (13,21) {};
\node[circle,inner sep = 0pt,minimum size =6pt,fill = black] (nx12y22) at (12,22) {};
\node[circle,inner sep = 0pt,minimum size =6pt,fill = black] (nx14y22) at (14,22) {};
\node[circle,inner sep = 0pt,minimum size =6pt,fill = black] (nx15y23) at (15,23) {};
\node[circle,inner sep = 0pt,minimum size =6pt,fill = black] (nx13y23) at (13,23) {};
\node[circle,inner sep = 0pt,minimum size =6pt,fill = black] (nx11y23) at (11,23) {};
\node[circle,inner sep = 0pt,minimum size = 3pt,fill = black] (nx7y3) at (7,3) {};
\node[circle,inner sep = 0pt,minimum size = 3pt,fill = black] (nx19y3) at (19,3) {};
\node[circle,inner sep = 0pt,minimum size = 3pt,fill = black] (nx17y3) at (17,3) {};
\node[circle,inner sep = 0pt,minimum size = 3pt,fill = black] (nx15y3) at (15,3) {};
\node[circle,inner sep = 0pt,minimum size = 3pt,fill = black] (nx11y3) at (11,3) {};
\node[circle,inner sep = 0pt,minimum size = 3pt,fill = black] (nx9y3) at (9,3) {};
\node[circle,inner sep = 0pt,minimum size = 3pt,fill = black] (nx11y5) at (11,5) {};
\node[circle,inner sep = 0pt,minimum size = 3pt,fill = black] (nx15y5) at (15,5) {};
\node[circle,inner sep = 0pt,minimum size = 3pt,fill = black] (nx17y9) at (17,9) {};
\node[circle,inner sep = 0pt,minimum size = 3pt,fill = black] (nx15y9) at (15,9) {};
\node[circle,inner sep = 0pt,minimum size = 3pt,fill = black] (nx11y9) at (11,9) {};
\node[circle,inner sep = 0pt,minimum size = 3pt,fill = black] (nx9y9) at (9,9) {};
\node[circle,inner sep = 0pt,minimum size = 3pt,fill = black] (nx7y11) at (7,11) {};
\node[circle,inner sep = 0pt,minimum size = 3pt,fill = black] (nx19y11) at (19,11) {};
\node[circle,inner sep = 0pt,minimum size = 3pt,fill = black] (nx17y11) at (17,11) {};
\node[circle,inner sep = 0pt,minimum size = 3pt,fill = black] (nx15y11) at (15,11) {};
\node[circle,inner sep = 0pt,minimum size = 3pt,fill = black] (nx11y11) at (11,11) {};
\node[circle,inner sep = 0pt,minimum size = 3pt,fill = black] (nx9y11) at (9,11) {};
\node[circle,inner sep = 0pt,minimum size = 3pt,fill = black] (nx15y13) at (15,13) {};
\node[circle,inner sep = 0pt,minimum size = 3pt,fill = black] (nx11y13) at (11,13) {};
\node[circle,inner sep = 0pt,minimum size = 3pt,fill = black] (nx19y17) at (19,17) {};
\node[circle,inner sep = 0pt,minimum size = 3pt,fill = black] (nx7y17) at (7,17) {};
\node[circle,inner sep = 0pt,minimum size = 3pt,fill = black] (nx17y17) at (17,17) {};
\node[circle,inner sep = 0pt,minimum size = 3pt,fill = black] (nx15y17) at (15,17) {};
\node[circle,inner sep = 0pt,minimum size = 3pt,fill = black] (nx11y17) at (11,17) {};
\node[circle,inner sep = 0pt,minimum size = 3pt,fill = black] (nx9y17) at (9,17) {};
\node[circle,inner sep = 0pt,minimum size = 3pt,fill = black] (nx19y19) at (19,19) {};
\node[circle,inner sep = 0pt,minimum size = 3pt,fill = black] (nx17y19) at (17,19) {};
\node[circle,inner sep = 0pt,minimum size = 3pt,fill = black] (nx15y21) at (15,21) {};
\node[circle,inner sep = 0pt,minimum size = 3pt,fill = black] (nx15y19) at (15,19) {};
\node[circle,inner sep = 0pt,minimum size = 3pt,fill = black] (nx11y19) at (11,19) {};
\node[circle,inner sep = 0pt,minimum size = 3pt,fill = black] (nx7y19) at (7,19) {};
\node[circle,inner sep = 0pt,minimum size = 3pt,fill = black] (nx19y25) at (19,25) {};
\node[circle,inner sep = 0pt,minimum size = 3pt,fill = black] (nx17y25) at (17,25) {};
\node[circle,inner sep = 0pt,minimum size = 3pt,fill = black] (nx15y25) at (15,25) {};
\node[circle,inner sep = 0pt,minimum size = 3pt,fill = black] (nx9y19) at (9,19) {};
\node[circle,inner sep = 0pt,minimum size = 3pt,fill = black] (nx11y21) at (11,21) {};
\node[circle,inner sep = 0pt,minimum size = 3pt,fill = black] (nx7y25) at (7,25) {};
\node[circle,inner sep = 0pt,minimum size = 3pt,fill = black] (nx9y25) at (9,25) {};
\node[circle,inner sep = 0pt,minimum size = 3pt,fill = black] (nx11y25) at (11,25) {};
\draw[black] (nx17y26) -- (nx13y28);
\draw[black] (nx13y28) -- (nx15y26);
\draw[black] (nx13y28) -- (nx11y26);
\draw[black] (nx9y26) -- (nx13y28);
\draw[black,very thick] (nx13y0) -- (nx17y2);
\draw[black,very thick] (nx13y0) -- (nx15y2);
\draw[black,very thick] (nx11y2) -- (nx13y0);
\draw[black,very thick] (nx9y2) -- (nx13y0);
\draw[black,very thick] (nx5y22) -- (nx7y19);
\draw[black,very thick] (nx7y25) -- (nx5y22);
\draw[black] (nx15y25) -- (nx15y26);
\draw[black] (nx11y26) -- (nx11y25);
\draw[black] (nx9y25) -- (nx9y26);
\draw[black,very thick] (nx9y9) -- (nx7y9);
\draw[black,very thick] (nx17y9) -- (nx19y9);
\draw[black,very thick] (nx9y3) -- (nx7y3);
\draw[black,very thick] (nx9y2) -- (nx9y3);
\draw[black,very thick] (nx11y3) -- (nx11y2);
\draw[black,very thick] (nx15y2) -- (nx15y3);
\draw[black,very thick] (nx17y3) -- (nx17y2);
\draw[black,very thick] (nx19y3) -- (nx17y3);
\draw[black] (nx17y3) -- (nx15y3);
\draw[black] (nx15y5) -- (nx17y3);
\draw[black] (nx15y3) -- (nx15y5);
\draw[black] (nx13y5) -- (nx15y3);
\draw[black,very thick] (nx11y3) -- (nx13y5);
\draw[black,very thick] (nx11y3) -- (nx11y5);
\draw[black,very thick] (nx9y3) -- (nx11y3);
\draw[black,very thick] (nx11y5) -- (nx9y3);
\draw[black,very thick] (nx11y5) -- (nx11y7);
\draw[black,very thick] (nx12y6) -- (nx11y5);
\draw[black] (nx13y5) -- (nx12y6);
\draw[black] (nx14y6) -- (nx13y5);
\draw[black] (nx15y5) -- (nx14y6);
\draw[black,very thick] (nx15y7) -- (nx15y5);
\draw[black,very thick] (nx14y6) -- (nx15y7);
\draw[black] (nx13y7) -- (nx14y6);
\draw[black] (nx12y6) -- (nx13y7);
\draw[black] (nx11y7) -- (nx12y6);
\draw[black,very thick] (nx15y7) -- (nx15y9);
\draw[black,very thick] (nx17y9) -- (nx15y7);
\draw[black,very thick] (nx15y9) -- (nx17y9);
\draw[black,very thick] (nx13y7) -- (nx15y9);
\draw[black] (nx11y9) -- (nx13y7);
\draw[black] (nx11y9) -- (nx9y9);
\draw[black] (nx11y7) -- (nx11y9);
\draw[black] (nx9y9) -- (nx11y7);
\draw[black,very thick] (nx17y9) -- (nx17y11);
\draw[black,very thick] (nx15y11) -- (nx15y9);
\draw[black,very thick] (nx11y9) -- (nx11y11);
\draw[black,very thick] (nx9y11) -- (nx9y9);
\draw[black,very thick] (nx9y11) -- (nx7y11);
\draw[black,very thick] (nx17y11) -- (nx19y11);
\draw[black] (nx17y11) -- (nx15y11);
\draw[black] (nx15y13) -- (nx17y11);
\draw[black] (nx15y11) -- (nx15y13);
\draw[black] (nx13y13) -- (nx15y11);
\draw[black,very thick] (nx11y11) -- (nx13y13);
\draw[black,very thick] (nx11y11) -- (nx11y13);
\draw[black,very thick] (nx9y11) -- (nx11y11);
\draw[black,very thick] (nx11y13) -- (nx9y11);
\draw[black,very thick] (nx15y13) -- (nx15y15);
\draw[black] (nx14y14) -- (nx15y13);
\draw[black] (nx13y13) -- (nx14y14);
\draw[black] (nx12y14) -- (nx13y13);
\draw[black,very thick] (nx11y13) -- (nx12y14);
\draw[black,very thick] (nx11y15) -- (nx11y13);
\draw[black] (nx12y14) -- (nx11y15);
\draw[black] (nx13y15) -- (nx12y14);
\draw[black] (nx14y14) -- (nx13y15);
\draw[black,very thick] (nx15y15) -- (nx14y14);
\draw[black,very thick] (nx17y17) -- (nx19y17);
\draw[black,very thick] (nx15y15) -- (nx15y17);
\draw[black,very thick] (nx17y17) -- (nx15y15);
\draw[black,very thick] (nx15y17) -- (nx17y17);
\draw[black,very thick] (nx13y15) -- (nx15y17);
\draw[black] (nx11y17) -- (nx13y15);
\draw[black] (nx11y15) -- (nx9y17);
\draw[black] (nx11y17) -- (nx11y15);
\draw[black] (nx9y17) -- (nx11y17);
\draw[black,very thick] (nx7y17) -- (nx9y17);
\draw[black,very thick] (nx17y19) -- (nx17y17);
\draw[black,very thick] (nx15y19) -- (nx15y17);
\draw[black,very thick] (nx11y19) -- (nx11y17);
\draw[black,very thick] (nx9y19) -- (nx9y17);
\draw[black,very thick] (nx17y19) -- (nx19y19);
\draw[black] (nx17y19) -- (nx15y19);
\draw[black] (nx15y21) -- (nx17y19);
\draw[black] (nx15y23) -- (nx15y21);
\draw[black] (nx15y21) -- (nx14y22);
\draw[black] (nx15y19) -- (nx15y21);
\draw[black] (nx13y21) -- (nx15y19);
\draw[black] (nx11y19) -- (nx13y21);
\draw[black] (nx11y19) -- (nx11y21);
\draw[black] (nx9y19) -- (nx11y19);
\draw[black,very thick] (nx7y19) -- (nx9y19);
\draw[black] (nx17y25) -- (nx19y25);
\draw[black] (nx17y26) -- (nx17y25);
\draw[black] (nx17y25) -- (nx15y25);
\draw[black] (nx15y23) -- (nx17y25);
\draw[black] (nx15y25) -- (nx15y23);
\draw[black] (nx13y23) -- (nx15y25);
\draw[black] (nx14y22) -- (nx15y23);
\draw[black] (nx13y21) -- (nx12y22);
\draw[black] (nx14y22) -- (nx13y21);
\draw[black] (nx13y23) -- (nx14y22);
\draw[black] (nx12y22) -- (nx13y23);
\draw[black] (nx12y22) -- (nx11y21);
\draw[black] (nx11y23) -- (nx12y22);
\draw[black] (nx11y21) -- (nx9y19);
\draw[black] (nx11y23) -- (nx11y21);
\draw[black] (nx11y25) -- (nx13y23);
\draw[black] (nx11y23) -- (nx11y25);
\draw[black] (nx9y25) -- (nx11y23);
\draw[black,very thick] (nx9y25) -- (nx7y25);
\draw[black] (nx11y25) -- (nx9y25);
\draw (13,0) node[anchor = north] {$A,B,I,Q$};
\draw (7,3) node[anchor = east] {$H$};
\draw (7,9) node[anchor = east] {$F$};
\draw (7,11) node[anchor = east] {$E$};
\draw (7,17) node[anchor = east] {$D$};
\draw (5,22) node[anchor = east] {$C, Q$};
\draw (19,3) node[anchor = west] {$I$};
\draw (19,9) node[anchor = west] {$H$};
\draw (19,11) node[anchor = west] {$F$};
\draw (19,17) node[anchor = west] {$E$};
\draw (19,19) node[anchor = west] {$D$};
\draw (19,25) node[anchor = west] {$C$};
\draw (13,28) node[anchor = south] {$A,B$};
\end{tikzpicture}
\end{center}
\legende{There is no solution to this edge-disjoint paths problem, even if $Q$ can cross other paths at every vertex. Edges contained in some tight cut are drawn thick.}
\label{fig:biflot-exclaim5}
\end{figure}

\begin{proof}
Suppose that there is a $(\{y_1,y_2\},X')$-no-path $Q$. Note that $Q$ can cross other paths at non-crossing vertices. We can take $\mathcal{P}$ uncrossed (uncrossing the paths does not change the set of edges used by the solution), and without loss of generality, $A$ is routed on the left of $B$.\\
Because of tight cuts, $M(1,1)$ contains $C$, $M(1,2)$ contains $E$, and $M(1,3)$ contains $H$. In $M(1,2)$, $E$ goes through $u_6$, ($u_3$ and $u_4$ are non-crossing vertices, and at least one of the paths $A$, $B$ or $D$ goes through these vertices) and uses $u_6c$ or $u_6u_8$. For the same reason, $E$ goes through $u_7$ and uses $u_7u_5$ or $u_7b$. Consider in $M(1,2)$ the cuts $C_1 := \{u_5u_7,u_5b,ab,ac,u_6c,u_6u_8\}$ and $C_2 := \{u_5u_7,u_7b,bd,cd,cu_8,u_6u_8\}$. There are exactly four paths routed through these two cuts, and $E$ must use three edges from at least one of them, because it can use neither $u_{10}d$ nor $au_3$. Then $F$ cannot go through $C_1$, proving that it can cross $A$ or $B$ only at vertex $d$ in $M(1,2)$.\\
Similarly, in $M(1,3)$, using path $H$ and the same two cuts (we call them $C'_1$ and $C'_2$), $F$ can only cross $A$ or $B$ at vertex $a$. Because of Lemma~\ref{lemme:uncross-biflot} and the tight cut between $M(1,2)$ and $M(1,3)$, $F$ crosses $A$ in $a$ of $M(1,3)$ and then crosses $B$ in $d$ of $M(1,2)$. As $F$ cannot use edges in both $C_1$ and $C'_2$, $F$ must traverse $d$ of $M(1,2)$ and $a$ of $M(1,3)$ from left to right or from right to left, an odd number of times. This is indeed a contradiction.
\end{proof}

We need a last ``local'' lemma, proving that without the presence of a $\lic$, the vertical paths of a column are shifted. This fills the holes of Lemma~\ref{lemme:xch-12}. As the no-paths only use three rows of each column (this is not fully proved yet), we will only need to study the case with a grid of $3$ rows. 

\begin{lemma}\label{lemme:biflot-exclaim6}
Let $G$ be a grid of $\xch$ of dimension $1 \times 3$. Then there is no set of edge-disjoint paths consisting of five $(Y,Y')$-paths, one $(x_1,x'_1)$-path and one $(x_2,x'_2)$-path.
\end{lemma}

\begin{figure}
\begin{center}
\begin{tikzpicture}[x=0.3cm,y=0.3cm,>=latex]
% Fichier d'origine : brouillon.pgf
\fill[black, nearly transparent, rounded corners=2pt] (14.5,1.5) rectangle (22.5,26.5);
\fill[black, nearly transparent, rounded corners=2pt] (3.5,1.5) rectangle (11.5,26.5);
\node[circle,inner sep = 0pt,minimum size = 3pt,fill = black] (nx11y7) at (11,7) {};
\node[circle,inner sep = 0pt,minimum size = 3pt,fill = black] (nx15y7) at (15,7) {};
\node[circle,inner sep = 0pt,minimum size = 3pt,fill = black] (nx15y15) at (15,15) {};
\node[circle,inner sep = 0pt,minimum size = 3pt,fill = black] (nx11y15) at (11,15) {};
\node[circle,inner sep = 0pt,minimum size = 3pt,fill = black] (nx15y23) at (15,23) {};
\node[circle,inner sep = 0pt,minimum size = 3pt,fill = black] (nx11y23) at (11,23) {};
\node[circle,inner sep = 0pt,minimum size = 3pt,fill = black] (nx17y2) at (17,2) {};
\node[circle,inner sep = 0pt,minimum size = 3pt,fill = black] (nx15y2) at (15,2) {};
\node[circle,inner sep = 0pt,minimum size = 3pt,fill = black] (nx11y2) at (11,2) {};
\node[circle,inner sep = 0pt,minimum size = 3pt,fill = black] (nx9y2) at (9,2) {};
\node[circle,inner sep = 0pt,minimum size = 3pt,fill = black] (nx17y26) at (17,26) {};
\node[circle,inner sep = 0pt,minimum size = 3pt,fill = black] (nx15y26) at (15,26) {};
\node[circle,inner sep = 0pt,minimum size = 3pt,fill = black] (nx11y26) at (11,26) {};
\node[circle,inner sep = 0pt,minimum size = 3pt,fill = black] (nx9y26) at (9,26) {};
\node[circle,inner sep = 0pt,minimum size = 3pt,fill = black] (nx7y9) at (7,9) {};
\node[circle,inner sep = 0pt,minimum size = 3pt,fill = black] (nx19y9) at (19,9) {};
\node[circle,inner sep = 0pt,minimum size =6pt,fill = black] (nx13y5) at (13,5) {};
\node[circle,inner sep = 0pt,minimum size =6pt,fill = black] (nx12y6) at (12,6) {};
\node[circle,inner sep = 0pt,minimum size =6pt,fill = black] (nx14y6) at (14,6) {};
\node[circle,inner sep = 0pt,minimum size =6pt,fill = black] (nx13y7) at (13,7) {};
\node[circle,inner sep = 0pt,minimum size =6pt,fill = black] (nx13y13) at (13,13) {};
\node[circle,inner sep = 0pt,minimum size =6pt,fill = black] (nx12y14) at (12,14) {};
\node[circle,inner sep = 0pt,minimum size =6pt,fill = black] (nx14y14) at (14,14) {};
\node[circle,inner sep = 0pt,minimum size =6pt,fill = black] (nx13y15) at (13,15) {};
\node[circle,inner sep = 0pt,minimum size =6pt,fill = black] (nx13y21) at (13,21) {};
\node[circle,inner sep = 0pt,minimum size =6pt,fill = black] (nx12y22) at (12,22) {};
\node[circle,inner sep = 0pt,minimum size =6pt,fill = black] (nx14y22) at (14,22) {};
\node[circle,inner sep = 0pt,minimum size =6pt,fill = black] (nx13y23) at (13,23) {};
\node[circle,inner sep = 0pt,minimum size = 3pt,fill = black] (nx7y3) at (7,3) {};
\node[circle,inner sep = 0pt,minimum size = 3pt,fill = black] (nx19y3) at (19,3) {};
\node[circle,inner sep = 0pt,minimum size = 3pt,fill = black] (nx17y3) at (17,3) {};
\node[circle,inner sep = 0pt,minimum size = 3pt,fill = black] (nx15y3) at (15,3) {};
\node[circle,inner sep = 0pt,minimum size = 3pt,fill = black] (nx11y3) at (11,3) {};
\node[circle,inner sep = 0pt,minimum size = 3pt,fill = black] (nx9y3) at (9,3) {};
\node[circle,inner sep = 0pt,minimum size = 3pt,fill = black] (nx11y5) at (11,5) {};
\node[circle,inner sep = 0pt,minimum size = 3pt,fill = black] (nx15y5) at (15,5) {};
\node[circle,inner sep = 0pt,minimum size = 3pt,fill = black] (nx17y9) at (17,9) {};
\node[circle,inner sep = 0pt,minimum size = 3pt,fill = black] (nx15y9) at (15,9) {};
\node[circle,inner sep = 0pt,minimum size = 3pt,fill = black] (nx11y9) at (11,9) {};
\node[circle,inner sep = 0pt,minimum size = 3pt,fill = black] (nx9y9) at (9,9) {};
\node[circle,inner sep = 0pt,minimum size = 3pt,fill = black] (nx7y11) at (7,11) {};
\node[circle,inner sep = 0pt,minimum size = 3pt,fill = black] (nx19y11) at (19,11) {};
\node[circle,inner sep = 0pt,minimum size = 3pt,fill = black] (nx17y11) at (17,11) {};
\node[circle,inner sep = 0pt,minimum size = 3pt,fill = black] (nx15y11) at (15,11) {};
\node[circle,inner sep = 0pt,minimum size = 3pt,fill = black] (nx11y11) at (11,11) {};
\node[circle,inner sep = 0pt,minimum size = 3pt,fill = black] (nx9y11) at (9,11) {};
\node[circle,inner sep = 0pt,minimum size = 3pt,fill = black] (nx15y13) at (15,13) {};
\node[circle,inner sep = 0pt,minimum size = 3pt,fill = black] (nx11y13) at (11,13) {};
\node[circle,inner sep = 0pt,minimum size = 3pt,fill = black] (nx19y17) at (19,17) {};
\node[circle,inner sep = 0pt,minimum size = 3pt,fill = black] (nx7y17) at (7,17) {};
\node[circle,inner sep = 0pt,minimum size = 3pt,fill = black] (nx17y17) at (17,17) {};
\node[circle,inner sep = 0pt,minimum size = 3pt,fill = black] (nx15y17) at (15,17) {};
\node[circle,inner sep = 0pt,minimum size = 3pt,fill = black] (nx11y17) at (11,17) {};
\node[circle,inner sep = 0pt,minimum size = 3pt,fill = black] (nx9y17) at (9,17) {};
\node[circle,inner sep = 0pt,minimum size = 3pt,fill = black] (nx19y19) at (19,19) {};
\node[circle,inner sep = 0pt,minimum size = 3pt,fill = black] (nx17y19) at (17,19) {};
\node[circle,inner sep = 0pt,minimum size = 3pt,fill = black] (nx15y21) at (15,21) {};
\node[circle,inner sep = 0pt,minimum size = 3pt,fill = black] (nx15y19) at (15,19) {};
\node[circle,inner sep = 0pt,minimum size = 3pt,fill = black] (nx11y19) at (11,19) {};
\node[circle,inner sep = 0pt,minimum size = 3pt,fill = black] (nx7y19) at (7,19) {};
\node[circle,inner sep = 0pt,minimum size = 3pt,fill = black] (nx19y25) at (19,25) {};
\node[circle,inner sep = 0pt,minimum size = 3pt,fill = black] (nx17y25) at (17,25) {};
\node[circle,inner sep = 0pt,minimum size = 3pt,fill = black] (nx15y25) at (15,25) {};
\node[circle,inner sep = 0pt,minimum size = 3pt,fill = black] (nx9y19) at (9,19) {};
\node[circle,inner sep = 0pt,minimum size = 3pt,fill = black] (nx11y21) at (11,21) {};
\node[circle,inner sep = 0pt,minimum size = 3pt,fill = black] (nx7y25) at (7,25) {};
\node[circle,inner sep = 0pt,minimum size = 3pt,fill = black] (nx9y25) at (9,25) {};
\node[circle,inner sep = 0pt,minimum size = 3pt,fill = black] (nx11y25) at (11,25) {};
\node[circle,inner sep = 0pt,minimum size = 3pt,fill = black] (nx22y14) at (22,14) {};
\node[circle,inner sep = 0pt,minimum size = 3pt,fill = black] (nx4y14) at (4,14) {};
\draw[black] (nx15y25) -- (nx15y26);
\draw[black] (nx11y26) -- (nx11y25);
\draw[black] (nx9y25) -- (nx9y26);
\draw[black] (nx9y9) -- (nx7y9);
\draw[black] (nx17y9) -- (nx19y9);
\draw[black] (nx9y3) -- (nx7y3);
\draw[black] (nx9y2) -- (nx9y3);
\draw[black] (nx11y3) -- (nx11y2);
\draw[black] (nx15y2) -- (nx15y3);
\draw[black] (nx17y3) -- (nx17y2);
\draw[black] (nx19y3) -- (nx17y3);
\draw[black] (nx17y3) -- (nx15y3);
\draw[black] (nx15y5) -- (nx17y3);
\draw[black] (nx15y3) -- (nx15y5);
\draw[black,thick] (nx13y5) -- (nx15y3);
\draw[black,thick] (nx11y3) -- (nx13y5);
\draw[black] (nx11y3) -- (nx11y5);
\draw[black] (nx9y3) -- (nx11y3);
\draw[black] (nx11y5) -- (nx9y3);
\draw[black] (nx11y5) -- (nx11y7);
\draw[black,thick] (nx12y6) -- (nx11y5);
\draw[black] (nx13y5) -- (nx12y6);
\draw[black] (nx14y6) -- (nx13y5);
\draw[black,thick] (nx15y5) -- (nx14y6);
\draw[black] (nx15y7) -- (nx15y5);
\draw[black,thick] (nx14y6) -- (nx15y7);
\draw[black] (nx13y7) -- (nx14y6);
\draw[black] (nx12y6) -- (nx13y7);
\draw[black,thick] (nx11y7) -- (nx12y6);
\draw[black] (nx15y7) -- (nx15y9);
\draw[black] (nx17y9) -- (nx15y7);
\draw[black] (nx15y9) -- (nx17y9);
\draw[black,thick] (nx13y7) -- (nx15y9);
\draw[black,thick] (nx11y9) -- (nx13y7);
\draw[black] (nx11y9) -- (nx9y9);
\draw[black] (nx11y7) -- (nx11y9);
\draw[black] (nx9y9) -- (nx11y7);
\draw[black] (nx17y9) -- (nx17y11);
\draw[black] (nx15y11) -- (nx15y9);
\draw[black] (nx11y9) -- (nx11y11);
\draw[black] (nx9y11) -- (nx9y9);
\draw[black] (nx9y11) -- (nx7y11);
\draw[black] (nx17y11) -- (nx19y11);
\draw[black] (nx17y11) -- (nx15y11);
\draw[black] (nx15y13) -- (nx17y11);
\draw[black] (nx15y11) -- (nx15y13);
\draw[black,thick] (nx13y13) -- (nx15y11);
\draw[black,thick] (nx11y11) -- (nx13y13);
\draw[black] (nx11y11) -- (nx11y13);
\draw[black] (nx9y11) -- (nx11y11);
\draw[black] (nx11y13) -- (nx9y11);
\draw[black] (nx15y13) -- (nx15y15);
\draw[black,thick] (nx14y14) -- (nx15y13);
\draw[black] (nx13y13) -- (nx14y14);
\draw[black] (nx12y14) -- (nx13y13);
\draw[black,thick] (nx11y13) -- (nx12y14);
\draw[black] (nx11y15) -- (nx11y13);
\draw[black,thick] (nx12y14) -- (nx11y15);
\draw[black] (nx13y15) -- (nx12y14);
\draw[black] (nx14y14) -- (nx13y15);
\draw[black,thick] (nx15y15) -- (nx14y14);
\draw[black] (nx17y17) -- (nx19y17);
\draw[black] (nx15y15) -- (nx15y17);
\draw[black] (nx17y17) -- (nx15y15);
\draw[black] (nx15y17) -- (nx17y17);
\draw[black,thick] (nx13y15) -- (nx15y17);
\draw[black,thick] (nx11y17) -- (nx13y15);
\draw[black] (nx11y15) -- (nx9y17);
\draw[black] (nx11y17) -- (nx11y15);
\draw[black] (nx9y17) -- (nx11y17);
\draw[black] (nx7y17) -- (nx9y17);
\draw[black] (nx17y19) -- (nx17y17);
\draw[black] (nx15y19) -- (nx15y17);
\draw[black] (nx11y19) -- (nx11y17);
\draw[black] (nx9y19) -- (nx9y17);
\draw[black] (nx17y19) -- (nx19y19);
\draw[black] (nx17y19) -- (nx15y19);
\draw[black] (nx15y21) -- (nx17y19);
\draw[black] (nx15y23) -- (nx15y21);
\draw[black,thick] (nx15y21) -- (nx14y22);
\draw[black] (nx15y19) -- (nx15y21);
\draw[black,thick] (nx13y21) -- (nx15y19);
\draw[black,thick] (nx11y19) -- (nx13y21);
\draw[black] (nx11y19) -- (nx11y21);
\draw[black] (nx9y19) -- (nx11y19);
\draw[black] (nx7y19) -- (nx9y19);
\draw[black] (nx17y25) -- (nx19y25);
\draw[black] (nx17y26) -- (nx17y25);
\draw[black] (nx17y25) -- (nx15y25);
\draw[black] (nx15y23) -- (nx17y25);
\draw[black] (nx15y25) -- (nx15y23);
\draw[black,thick] (nx13y23) -- (nx15y25);
\draw[black,thick] (nx14y22) -- (nx15y23);
\draw[black] (nx13y21) -- (nx12y22);
\draw[black] (nx14y22) -- (nx13y21);
\draw[black] (nx13y23) -- (nx14y22);
\draw[black] (nx12y22) -- (nx13y23);
\draw[black,thick] (nx12y22) -- (nx11y21);
\draw[black,thick] (nx11y23) -- (nx12y22);
\draw[black] (nx11y21) -- (nx9y19);
\draw[black] (nx11y23) -- (nx11y21);
\draw[black,thick] (nx11y25) -- (nx13y23);
\draw[black] (nx11y23) -- (nx11y25);
\draw[black] (nx9y25) -- (nx11y23);
\draw[black] (nx9y25) -- (nx7y25);
\draw[black] (nx11y25) -- (nx9y25);
\draw[black] (nx7y25) -- (nx4y14);
\draw[black] (nx4y14) -- (nx7y19);
\draw[black] (nx7y17) -- (nx4y14);
\draw[black] (nx4y14) -- (nx7y11);
\draw[black] (nx4y14) -- (nx7y9);
\draw[black] (nx4y14) -- (nx7y3);
\draw[black] (nx19y3) -- (nx22y14);
\draw[black] (nx22y14) -- (nx19y9);
\draw[black] (nx19y11) -- (nx22y14);
\draw[black] (nx22y14) -- (nx19y17);
\draw[black] (nx19y19) -- (nx22y14);
\draw[black] (nx22y14) -- (nx19y25);
\draw (22,26) node {$R$};
\draw (4,26) node {$L$};
\draw (22,14) node[anchor = west] {$s'$};
\draw (4,14) node[anchor = east] {$s$};
\draw (11,2) node[anchor = north] {$x'_2$};
\draw (9,2) node[anchor = north] {$x'_1$};
\draw (11,26) node[anchor = south] {$x_2$};
\draw (9,26) node[anchor = south] {$x_1$};
\end{tikzpicture}
\end{center}
\legende{There is no solution to this edge-disjoint paths problem, with $5$ paths between $s$ and $s'$, see Lemma~\ref{lemme:biflot-exclaim6}.}
\label{fig:biflot-exclaim6}
\end{figure}

\begin{proof}
Suppose that these paths exist. We distinguish two special cuts $\delta(L)$ and $\delta(R)$. There are exactly $12$ vertices for crossings, and $10$ are needed. The $(Y,Y')$-paths use $5$ edges of $\delta(L)$ and $5$ of $\delta(R)$. The $(x_1,x'_1)$-path $A$ and the $(x_2,x'_2)$-path $B$ both use an even number of edges in these two cuts. Moreover, they can do at most $2$ crossings in each of the three groups of four crossing vertices (corresponding to the crossing vertices of an $\xch$ graph), thus they go through each of these groups. Then each uses at least $6$ edges in the two cuts, and because of parity, $A$ uses $4$ edges of $\delta(L)$ and $2$ of $\delta(R)$, and $B$ uses $2$ edges of $\delta(L)$ and $4$ edges of $\delta(R)$. Because $d(L) = d(R) = 12$, there cannot be more. Thus, each $(Y,Y')$-path uses exactly one edge of $\delta(R)$, and one of $\delta(L)$. Then, in the central $\xch$ graph, there are exactly one edge of $\delta(L)$ and one edge of $\delta(R)$ used by $A$, and two edges of $\delta(R)$ used by $B$. At least two edges of $\delta(L)$ must be used by the $(Y,Y')$-paths, and at most one of $\delta(R)$. But this leads to a contradiction, as each $(Y,Y')$-path cannot use more than one edge in any of the two cuts.
\end{proof}

\section{Reduction}\label{sec:reduction}

Let $\varphi$ be a Boolean formula in conjunctive normal form, every clause is composed of $3$ literals, with $n \geq 3$ clauses over a set of $p'\geq 3$ variables. We encode the formula in a graph $G_{\varphi}$, a grid of $\xch$ and $\lic$ with $n$ columns. We need two rows for each variable, plus between every two of these rows, and after the last one, a buffer consisting of $q = 4(p'+3)n+2$ rows. Then, the grid has exactly $p = 2p'(q+1)$ rows.\\
The $i^{\textit{th}}$ variable $X_i$ corresponds to rows $1 + 2(i-1)(1+q)$ and $q + 2 + 2(i-1)(1+q)$. Intuitively, one no-path will be routed through one of these two rows (but actually, we can only enforce that it will stay near one of them). All the other rows are parts of buffers, and then contain only $\xch$. Row $1 + 2(i-1)(1+q)$ corresponds to the assignment of \emph{true} to $X_i$, so there is a $\lic$ in each column corresponding to the clause where $X_i$ appears positively. Similarly, row $q + 2 + 2(i-1)(1+q)$ corresponds to the assignment of \emph{false} to $X_i$, so there is a $\lic$ in each column corresponding to the clause where $X_i$ appears negatively. All other subgraphs are $\xch$. In summary, $M(i,j)$ is a $\lic$ if and only if :
\begin{itemize}
\item either $j = 1 + 2(i-1)(1+q)$ and $X_i$ appears positively in the $j^{\textit{th}}$ clause, 
\item or $j = q + 2 + 2(i-1)(1+q)$ and $X_i$ appears negatively in the $j^{\textit{th}}$ clause.
\end{itemize}
We add two terminals for the vertical paths, $x$ and $x'$. As the number of rows is even, and we want the vertical paths to be shifted an odd number of times, we add the following edges : $xx_{4k+1}$, $xx_{4k+2}$, $x'_{4k+3}x'$ and $x'_{4k+4}x'$ for each $k \in \ival{0}{n-1}$. Thus we enforce the parity of the number of shifts. Moreover, to reduce the number of odd vertices, we add the edges $x_{4k+3}x_{4k+4}$ and $x'_{4k+1}x'_{4k+2}$ for each $k \in \ival{0}{n-1}$.\\
For each variable $X_i$, $i \in \ival{1}{p'}$, we add two new vertices $w_i$ and $w'_i$, and the edges $w_iy_j$ and $w'_iy'_j$ for all $j \in \ival{4(i-1)(q+1)+1}{4(i-1)(q+1)+2q+4}$ (that is $w_i$ is connected to the vertices of $Y$ between the two rows encoding $X_i$). These new vertices will be the extremities of the no-paths. Then we add two terminals for the horizontal paths $y$ and $y'$. $y$ is connected to $w_i$ by $2q+3$ parallel edges, and to every vertex of $Y$ that still have a degree one by a single edge. Symmetrically, $y'$ is connected to $w'_i$ by $2q+3$ parallel edges, and to every vertex of $Y'$ not adjacent to some $w'_i$. Thus $y$ and $y'$ have degrees $2p - p'$.\\
Finally we demand to find $2n$ $(x,x')$-paths and $2p - p'$ $(y,y')$-paths, defining the demand graph $H_{\varphi}$. The vertices of odd degree in $G_{\varphi}+H_{\varphi}$ are the vertices of $W := \{w_i~:~i \in \ival{1}{p'}\}$ and $W' := \{w'_i~:~i \in \ival{1}{p'}\}$. The following cuts are tight : $\delta(x)$, $\delta(x')$, $\delta(y)$ and $\delta(y')$. Figure~\ref{fig:biflot-ligne1} gives an overview of the reduction.\\
We define precisely the \emph{vertical paths} to be the $(x,x')$-path of a solution, and the \emph{horizontal paths} are the $(y,y')$-paths. We say that two paths are \emph{parallel} if they are both horizontal, or both vertical, otherwise they are \emph{orthogonal}.

\begin{figure}[p]
\begin{center}
\scriptsize

\end{center}
\thispagestyle{empty}
\legende{The first $2(q+1)$ rows of the graph obtained by reduction. These are the rows corresponding to the encoding of the first variable.}
\label{fig:biflot-ligne1}
\end{figure}

\subsection{No-paths}

Let $\mathcal{P}$ be a solution to the edge-disjoint paths problem $(G_{\varphi},H_{\varphi})$. By considering $\mathcal{P}$ as a set of cycles of $G_{\varphi}+H_{\varphi}$ containing exactly one edge from $E(H_{\varphi})$, the vertices of $G_{\varphi}$ have the same degrees in $G_{\varphi}+H_{\varphi}$ and $G_{\varphi}+H_{\varphi} \setminus E(\mathcal{P})$. So the edges of the complementary graph of the solution $\mathcal{P}$ is a $W \cup W'$-join $Q \subset E(G_{\varphi})$. The following lemma proves that the no-paths cannot traverse the buffers, implying that the vertical paths stay in their respective columns. This last fact will allow us to apply the previous lemmas.

\begin{lemma}\label{mainlemma}
Let $G$ be a grid of dimension $n \times p$.
Let $\mathcal{P}$ be an uncrossed set of $(X,X')$-paths and $(Y,Y')$-paths pairwise edge-disjoint. Suppose there exists $i \in \ival{2}{p-2}$ such that for all $j \in \ival{1}{n}$, $M(i,j)$ and $M(i+1,j)$ are $\xch$ and there are exactly four crossings of paths of $\mathcal{P}$ in $M(i,j)$ and in $M(i+1,j)$.
Then, there is no no-path between a vertex of row $i-1$ and a vertex of row $i+2$.
\end{lemma}

\begin{proof}

Let $i \in \ival{2}{p-2}$ such that each crossing vertex of rows $i$ and $i+1$ is used for a crossing, and suppose there is a no-path $Q$ between rows $i-1$ and $i+2$. Then $Q$ must pass through $u_5u_7$ or $u_6u_8$ in some $M(j,i)$, say $u_6u_8$ by symmetry, and then passes through $e^{j,i}_3$, $e^{j,i}_4$, $e^{j+1,i}_1$ or $e^{j+1,i}_2$ (because all the $a$, $b$, $c$ and $d$ vertices are used by paths). \\
 Let $P_1$, $P_2$, $P_3$, $P_4$ be the paths going through edges $c^{j,i}u^{j,i}_8$, $d^{j,i}u^{j,i}_{11}$, $a^{j,i+1}u^{j,i+1}_3$, $c^{j,i+1}u^{j,i+1}_6$ respectively. We consider the cut $C = \delta(U)$, with $U = \{u^{j,i}_8, u^{j,i}_{11}, u^{j,i}_{12}, u^{j,i+1}_3, u^{j,i+1}_4, u^{j,i+1}_6\}$, intersected by these four paths and $Q$. Because $|C| = 8$, at most four different paths and no-paths intersect $C$. As the solution is uncrossed, $P_1$ and $P_2$ are different and parallel, and $P_3$ and $P_4$ also. There is no crossing vertex among $U$, thus $P_2 = P_3$. By Lemma~\ref{lemme:uncrossing}, $P_1$, $P_2$ and $P_4$ are crossed by their orthogonal paths in the same order, then $P_1 = P_4$. We remark that $j \neq n$, otherwise the routing is not possible.\\
At least one of $P_1$, $P_2$ and $Q$ must go through $U' = \{u^{j+1,i}_7, u^{j+1,i}_9, u^{j+1,i}_{10}, u^{j+1,i+1}_1, u^{j+1,i+1}_2, u^{j+1,i+1}_5\}$. By a similar argument, there are paths $P'_1$ using $b^{j+1,i}u^{j+1,i}_7$ and $u_5^{j+1,i+1}b^{j+1,i+1}$ and $P'_2$ using $d^{j+1,i}u^{j+1,i}_{10}$ and $u_2^{j+1,i+1}a^{j+1,i+1}$. Exactly one of the four considered paths must use two of the edges $u^{j,i}_6u^{j,i}_8$, $u^{j+1,i}_5u^{j+1,i}_7$, $u^{j,i+1}_6u^{j,i+1}_8$ and $u^{j+1,i+1}_5u^{j+1,i+1}_7$, say $P_1$ (the other cases are isomorphic). Then, we can consider the cycle consisting of the subpath of $P_1$ between $c^{j,i}$ and $c^{j,i+1}$, the edges $c^{j,i+1}a^{j,i+1}$ and $c^{j,i}d^{j,i}$ and the subpath of $P_2$ between $d^{j,i}$ and $a^{j,i+1}$. Because of Lemma~\ref{lemme:uncrossing}, no path can enter inside this cycle, but it contains at least one crossing vertex in rows $i$ and $i+1$, leading to a contradiction.

\end{proof}

\begin{lemma}\label{lemme:non-chemins}
The complementary $Q$ of $\mathcal{P}$ can be decomposed into cycles and one $(w_i,w'_i)$-path for each $i \in \ival{1}{p'}$.
\end{lemma}

\begin{proof}
There are exactly $np$ gadgets in the graph, among which exactly $3n$ are $\lic$. The number of crossing vertices is exactly $4np + 6n$. Moreover, there are $2n$ vertical paths, crossing each of the $2p - p'$ horizontal paths. At most $2n(p'+3)$ crossing vertices are not used to cross paths. As the number of rows in a buffer is $q = 4(p'+3)n+2$, there are at least two consecutive rows where all the crossing vertices are used to cross paths. Then we can apply Lemma~\ref{mainlemma} : there is no no-path going from the top to the bottom of a buffer.\\
Because of the parity of vertices, $G \setminus \mathcal{P}$ is a $W \cup W'$-join, and can be decomposed in cycles and some paths with extremities in $W \cup W'$, but there can only be $(w_i,w'_i)$-paths, as all other possible paths would have to go through a buffer.
\end{proof}

This leads to a key consequence :

\begin{lemma}\label{lemme:verticalpaths}
Every vertical path is contained in one column.
\end{lemma}

\begin{proof}
$2p - p'$ horizontal paths and $p'$ no-paths are routed through each vertical cut, which contains exactly $2p$ edges. Then vertical paths cannot use edges of vertical cuts.
\end{proof}

\subsection{Proof of the reduction}

\begin{theorem}\label{th:biflot}
The planar edge-disjoint paths problem is strongly NP-complete, even if the demand graph has only two edges, with terminals lying on the boundary of the infinite face of the input graph.
\end{theorem}

\begin{proof}
We use the graph of polynomial size built in Section~\ref{sec:reduction}.\\
Suppose that the formula is satisfiable, and consider an assignment satisfying $\varphi$. We route two horizontal paths through each row with the following exceptions, where we route only one horizontal path :
\begin{itemize}
\item in row $1 +2(k-1)(1+q)$ if the value of variable $k$ is \emph{true},
\item in row $q + 2 + 2(k-1)(1+q)$ if the value of variable $k$ is \emph{false}
\end{itemize}
Then, for each column, we switch the two vertical paths, using Figure~\ref{fig:biflot-xch-22}, except in the row corresponding to the first variable satisfying the clause associated with the column, where we keep the vertical paths. This is possible by Lemma~\ref{lemme:lic-12}, as there is a $\lic$ at this intersection. The vertical paths of each column are kept exactly once. Then these paths are valid, whence the edge-disjoint paths problem has a solution.\\
Suppose now that there is a solution to the edge-disjoint paths problem. By Lemma~\ref{lemme:verticalpaths}, the vertical paths do not intersect vertical cuts, and each no-path or horizontal path intersects each vertical path only once. We show that each no-path can use at most $3$ distinct rows in each column (and thus are separated by at least $q - 2n-1$ rows), iteratively on the columns of the grid. This is done by applying Lemma~\ref{lemme:decalage-nc} (because there is no other no-path in the two nearest rows). Then, for each $k \in \ival{1}{p'}$, the $(w_k,w'_k)$-path cannot intersect both the rows between $2(q+1)(k-1)-2$ and $2(q+1)(k-1)+4$, and the rows between  $2(q+1)(k-1) + q - 1$ and $2(q+1)(k-1) + q + 5$, as $q > 2n + 5$. If it intersects the first group, we set the variable $V_k$ to \emph{true}, otherwise we set it to \emph{false}.\\
We consider an arbitrary column. In each row except for $p'$ distinct groups of three consecutive rows, there are two horizontal paths traversing the gadget from left to right, as in the hypothesis of Lemmas~\ref{lemme:xch-22} and~\ref{lemme:lic-22}. Moreover, by construction, there are $\xch$ above and below every $\lic$, so we can effectively apply Lemmas~\ref{lemme:xch-22} and~\ref{lemme:lic-22} to all the gadgets, except a block of five consecutive gadgets for each variable, which contains the intersection of the corresponding no-path with the given column.\\
Now, for each of these blocks, if there are only $\xch$, the first and fifth gadgets follow Lemma~\ref{lemme:xch-22}, and by Lemma~\ref{lemme:biflot-exclaim6}, the vertical paths cannot be shifted here. Otherwise, there is a $\lic$ and a no-path in one of the three nearest rows, thus the variable assignment validates the corresponding clause (the fact that the no-path is not forced to pass through the $\lic$ has no consequence, as long as it must be close enough). Then whenever the vertical paths are not shifted, the clause associated with the column is satisfied. As each column sees its vertical paths kept at least once, each clause is satisfied, thus the assignment is feasible for $\varphi$.
\end{proof}

\section{Directed case}\label{section:directed}

Using a folkloric reduction, the following result is an obvious consequence of Theorem~\ref{th:biflot}.

\begin{corollary}
The arc-disjoint paths problem is strongly NP-complete, even if $G$ is planar, the demand graph has only two arcs and the terminals lie on the boundary of the infinite face of $G$.
\end{corollary}

% We can do a little better.

% \begin{corollary}\label{cor:1flow2paths}
% The planar arc-disjoint path problem is NP-complete, even if the offer graph is planar, the demand graph $G$ has only two terminals $s$ and $t$ lying on the infinite face of $G$, and with only two demands for $st$ (but possibly many demands for $ts$).
% \end{corollary}

% \begin{proof}
% We slightly modify the reduction of the proof of theorem~\ref{th:biflot}. First, we delete $x$ and $x'$. We orient the edges of each vertical cut from right to left. We suppose without loss of generality that $n$ is odd. For all $i \in \ival{1}{(n-1)/2}$, we add the following arcs : $(x_{8i- 3},x_{8i+2})$, $(x_{8i- 2},x_{8i+1})$, and for all $i \in \ival{1}{(n-1)/2}$, $(x'_{8i- 5},x'_{8i})$ and $(x'_{8i- 4},x'_{8i-1})$. Then we had two new vertices $s'$ and $t'$. There are $2p'-p$ arcs from $t'$ to $t$ and arcs $(x'_{4n},t')$, $(x'_{4n-1},t')$, plus $2p-p'$ arcs from $s$ to $s'$, and arcs $(s',x_1)$ and $(s',x_2)$. The demand are $2$ from $s'$ to $t'$, and $2p-p'$ from $t'$ to $s'$.\\
% The reduction is correct, as the two $(s',t')$-paths must take the role of the vertical paths, because they must use the new edges between vertices of $X$ and of $X'$ (they belong to tight cuts).\qed
% \end{proof}

We use again the grids of subgraphs, but with different subgraphs, to prove the following theorem :

\begin{theorem}\label{th:diflot}
The planar arc-disjoint paths problem is NP-complete, even if $G$ is acyclic and $H$ consists of two sets of parallel edges.
\end{theorem}

\begin{figure}
\centering
\begin{tikzpicture}[x=0.4cm,y=0.4cm,>=latex]
% Fichier d'origine : biflot-yesno.pgf
\node[circle,inner sep = 0pt,minimum size =5pt,fill = black] (nx19y18) at (19,18) {};
\node[circle,inner sep = 0pt,minimum size =5pt,fill = black] (nx15y18) at (15,18) {};
\node[circle,inner sep = 0pt,minimum size =5pt,fill = black] (nx19y24) at (19,24) {};
\node[circle,inner sep = 0pt,minimum size =5pt,fill = black] (nx15y24) at (15,24) {};
\node[circle,inner sep = 0pt,minimum size =5pt,fill = black] (nx21y21) at (21,21) {};
\node[circle,inner sep = 0pt,minimum size =5pt,fill = black] (nx19y22) at (19,22) {};
\node[circle,inner sep = 0pt,minimum size =5pt,fill = black] (nx19y20) at (19,20) {};
\node[circle,inner sep = 0pt,minimum size =5pt,fill = black] (nx17y20) at (17,20) {};
\node[circle,inner sep = 0pt,minimum size =5pt,fill = black] (nx17y22) at (17,22) {};
\node[circle,inner sep = 0pt,minimum size =5pt,fill = black] (nx15y20) at (15,20) {};
\node[circle,inner sep = 0pt,minimum size =5pt,fill = black] (nx15y22) at (15,22) {};
\node[circle,inner sep = 0pt,minimum size =5pt,fill = black] (nx13y21) at (13,21) {};
\node[circle,inner sep = 0pt,minimum size =5pt,fill = black] (nx7y19) at (7,19) {};
\node[circle,inner sep = 0pt,minimum size =5pt,fill = black] (nx9y21) at (9,21) {};
\node[circle,inner sep = 0pt,minimum size =5pt,fill = black] (nx7y23) at (7,23) {};
\node[circle,inner sep = 0pt,minimum size =5pt,fill = black] (nx7y21) at (7,21) {};
\node[circle,inner sep = 0pt,minimum size =5pt,fill = black] (nx5y19) at (5,19) {};
\node[circle,inner sep = 0pt,minimum size =5pt,fill = black] (nx5y23) at (5,23) {};
\node[circle,inner sep = 0pt,minimum size =5pt,fill = black] (nx5y21) at (5,21) {};
\node[circle,inner sep = 0pt,minimum size =5pt,fill = black] (nx3y21) at (3,21) {};
\draw[->,black] (nx15y24) -- (nx15y22);
\draw[->,black] (nx19y24) -- (nx19y22);
\draw[->,black] (nx19y22) -- (nx19y20);
\draw[->,black] (nx19y20) -- (nx21y21);
\draw[->,black] (nx19y20) -- (nx19y18);
\draw[->,black] (nx17y20) -- (nx19y20);
\draw[->,black] (nx19y22) -- (nx17y22);
\draw[->,black] (nx17y20) -- (nx15y20);
\draw[->,black] (nx17y22) -- (nx17y20);
\draw[->,black] (nx15y22) -- (nx17y22);
\draw[->,black] (nx15y20) -- (nx15y18);
\draw[->,black] (nx15y22) -- (nx15y20);
\draw[->,black] (nx13y21) -- (nx15y22);
\draw[->,black] (nx5y21) -- (nx5y19);
\draw[->,black] (nx7y21) -- (nx7y19);
\draw[->,black] (nx7y23) -- (nx7y21);
\draw[->,black] (nx7y21) -- (nx9y21);
\draw[->,black] (nx5y21) -- (nx7y21);
\draw[->,black] (nx5y23) -- (nx5y21);
\draw[->,black] (nx3y21) -- (nx5y21);
\draw (19,18) node[anchor = north] {$c'$};
\draw (15,18) node[anchor = north] {$b'$};
\draw (21,21) node[anchor = west] {$a'$};
\draw (19,24) node[anchor = south] {$c$};
\draw (15,24) node[anchor = south] {$b$};
\draw (13,21) node[anchor = east] {$a$};
\draw (7,19) node[anchor = north] {$c'$};
\draw (5,19) node[anchor = north] {$b'$};
\draw (9,21) node[anchor = west] {$a'$};
\draw (7,23) node[anchor = south ] {$c$};
\draw (5,23) node[anchor = south ] {$b$};
\draw (3,21) node[anchor = east] {$a$};
\draw (17,16) node {Graph $\textrm{YES}$};
\draw (6,16) node {Graph $\textrm{NO}$};
\end{tikzpicture}
\legende{In $\textrm{NO}$, there is no path from $c$ to $b'$, whereas it is possible in $\textrm{YES}$, as long as no other path goes through the graph.}
\label{figyesandno}
\end{figure}

We will reduce from {\sc Satisfiability}. Let $C_1 \land \ldots \land C_n$ be a formula with $n$ clauses, over the set of variables $\{X_1,\ldots, X_p\}$. Let $G_1$ be a grid with $n$ columns and $2p$ rows, where each point $G_1(i,j), i \in \ival{1}{2p}, j \in \ival{1}{n}$ of the grid is a special subgraph, defined as follows (see figure~\ref{figyesandno}):
\begin{itemize}
\item $G_1(2i-1,j)$ is the graph $\textrm{YES}$ if $X_i$ appears positively in $C_j$,
\item $G_1(2i,j)$ is the graph $\textrm{YES}$ if $X_i$ appears negatively in $C_j$,
\item $G_1(i,j)$ is $\textrm{NO}$ in all other cases.
\end{itemize}

\setcounter{myclaim}{0}
\begin{myclaim}\label{claimfirst}
The formula is satisfiable if and only if there is a set $\mathcal{P}$ of arc-disjoint paths in $G_1$ such that :
\begin{itemize}
\item[(i)] for each $j \in \ival{1}{n}$, there is a path $P_j$ in $\mathcal{P}$ from $c \in G_1(1,j)$ to $b' \in G_1(2p,j)$,
\item[(ii)] for each $i \in \ival{1}{p}$, there is a path $Q_i$ in $\mathcal{P}$ either from $a \in G_1(2i-1,1)$ to $a' \in G_1(2i-1,n)$ or from $a \in G_1(2i,1)$ to $a' \in G_1(2i,n)$.
\end{itemize}
\end{myclaim}

\begin{proof}
Suppose that $\mathcal{P}$ exists. For all $i \in \ival{1}{p}$, if $Q_i$ has extremities $a \in G_1(2i-1,1)$ and $a' \in G_1(2i-1,n)$, then assign value \emph{false} to $X_i$, otherwise assign value \emph{true}. Horizontal cuts and vertical cuts are directed, thus every path $Q_i$, $(i \in \ival{1}{p})$ is contained in a single row and every path $P_j, j \in \ival{1}{n}$, is contained in a single column. For each path $P_j, j \in \ival{1}{n}$, let $i$ be the index of the first row where $P_j$ goes through the left edge between $G_1(i,j)$ and $G_1(i+1,j)$. Then $P_j$ is the only path that goes through $G_1(i,j)$, and $G_1(i,j)$ is a $\textrm{YES}$ graph. If $i$ is even, it means that $X_{\frac{i}{2}}$ appears negatively in $C_j$ and this variable has value \emph{false}, thus $C_j$ is satisfied. Otherwise $X_{\frac{i+1}{2}}$ appears positively in $C_j$ and the value of this variable is \emph{true}, thus $C_j$ is also satisfied. Then the formula is satisfied. The converse is obvious.
\end{proof}

We just have to enforce paths to be as required in the previous claim. Condition $(i)$ is easy to satisfy. To check condition $(ii)$ we need some gadgets, depicted in Figure~\ref{figrouteurs}.

\begin{figure}
\centering
\begin{tikzpicture}[x=0.4cm,y=0.4cm,>=latex]
% Fichier d'origine : biflot-routeurs.pgf
\node[circle,inner sep = 0pt,minimum size =3pt,fill = black] (nx28y8) at (28,8) {};
\node[circle,inner sep = 0pt,minimum size =3pt,fill = black] (nx24y6) at (24,6) {};
\node[circle,inner sep = 0pt,minimum size =3pt,fill = black] (nx26y8) at (26,8) {};
\node[circle,inner sep = 0pt,minimum size =3pt,fill = black] (nx26y10) at (26,10) {};
\node[circle,inner sep = 0pt,minimum size =3pt,fill = black] (nx24y8) at (24,8) {};
\node[circle,inner sep = 0pt,minimum size =3pt,fill = black] (nx24y10) at (24,10) {};
\node[circle,inner sep = 0pt,minimum size =3pt,fill = black] (nx20y8) at (20,8) {};
\node[circle,inner sep = 0pt,minimum size =3pt,fill = black] (nx22y10) at (22,10) {};
\node[circle,inner sep = 0pt,minimum size =3pt,fill = black] (nx24y12) at (24,12) {};
\node[circle,inner sep = 0pt,minimum size =3pt,fill = black] (nx22y12) at (22,12) {};
\node[circle,inner sep = 0pt,minimum size =3pt,fill = black] (nx20y10) at (20,10) {};
\node[circle,inner sep = 0pt,minimum size =3pt,fill = black] (nx10y6) at (10,6) {};
\node[circle,inner sep = 0pt,minimum size =3pt,fill = black] (nx14y8) at (14,8) {};
\node[circle,inner sep = 0pt,minimum size =3pt,fill = black] (nx12y8) at (12,8) {};
\node[circle,inner sep = 0pt,minimum size =3pt,fill = black] (nx12y10) at (12,10) {};
\node[circle,inner sep = 0pt,minimum size =3pt,fill = black] (nx10y8) at (10,8) {};
\node[circle,inner sep = 0pt,minimum size =3pt,fill = black] (nx8y8) at (8,8) {};
\node[circle,inner sep = 0pt,minimum size =3pt,fill = black] (nx6y6) at (6,6) {};
\node[circle,inner sep = 0pt,minimum size =3pt,fill = black] (nx6y8) at (6,8) {};
\node[circle,inner sep = 0pt,minimum size =3pt,fill = black] (nx8y10) at (8,10) {};
\node[circle,inner sep = 0pt,minimum size =3pt,fill = black] (nx8y12) at (8,12) {};
\node[circle,inner sep = 0pt,minimum size =3pt,fill = black] (nx6y10) at (6,10) {};
\node[circle,inner sep = 0pt,minimum size =3pt,fill = black] (nx2y8) at (2,8) {};
\node[circle,inner sep = 0pt,minimum size =3pt,fill = black] (nx4y10) at (4,10) {};
\node[circle,inner sep = 0pt,minimum size =3pt,fill = black] (nx2y10) at (2,10) {};
\node[circle,inner sep = 0pt,minimum size =3pt,fill = black] (nx4y12) at (4,12) {};
\node[circle,inner sep = 0pt,minimum size =3pt,fill = black] (nx4y14) at (4,14) {};
\node[circle,inner sep = 0pt,minimum size =3pt,fill = black] (nx24y16) at (24,16) {};
\node[circle,inner sep = 0pt,minimum size =3pt,fill = black] (nx24y18) at (24,18) {};
\node[circle,inner sep = 0pt,minimum size =3pt,fill = black] (nx22y18) at (22,18) {};
\node[circle,inner sep = 0pt,minimum size =3pt,fill = black] (nx26y20) at (26,20) {};
\node[circle,inner sep = 0pt,minimum size =3pt,fill = black] (nx24y20) at (24,20) {};
\node[circle,inner sep = 0pt,minimum size =3pt,fill = black] (nx24y22) at (24,22) {};
\node[circle,inner sep = 0pt,minimum size =3pt,fill = black] (nx26y24) at (26,24) {};
\node[circle,inner sep = 0pt,minimum size =3pt,fill = black] (nx24y24) at (24,24) {};
\node[circle,inner sep = 0pt,minimum size =3pt,fill = black] (nx22y22) at (22,22) {};
\node[circle,inner sep = 0pt,minimum size =3pt,fill = black] (nx20y22) at (20,22) {};
\node[circle,inner sep = 0pt,minimum size =3pt,fill = black] (nx22y24) at (22,24) {};
\node[circle,inner sep = 0pt,minimum size =3pt,fill = black] (nx24y28) at (24,28) {};
\node[circle,inner sep = 0pt,minimum size =3pt,fill = black] (nx22y28) at (22,28) {};
\node[circle,inner sep = 0pt,minimum size =3pt,fill = black] (nx22y26) at (22,26) {};
\node[circle,inner sep = 0pt,minimum size =3pt,fill = black] (nx20y26) at (20,26) {};
\node[circle,inner sep = 0pt,minimum size =3pt,fill = black] (nx18y26) at (18,26) {};
\node[circle,inner sep = 0pt,minimum size =3pt,fill = black] (nx14y22) at (14,22) {};
\node[circle,inner sep = 0pt,minimum size =3pt,fill = black] (nx14y20) at (14,20) {};
\node[circle,inner sep = 0pt,minimum size =3pt,fill = black] (nx12y18) at (12,18) {};
\node[circle,inner sep = 0pt,minimum size =3pt,fill = black] (nx12y20) at (12,20) {};
\node[circle,inner sep = 0pt,minimum size =3pt,fill = black] (nx10y20) at (10,20) {};
\node[circle,inner sep = 0pt,minimum size =3pt,fill = black] (nx8y18) at (8,18) {};
\node[circle,inner sep = 0pt,minimum size =3pt,fill = black] (nx8y20) at (8,20) {};
\node[circle,inner sep = 0pt,minimum size =3pt,fill = black] (nx10y22) at (10,22) {};
\node[circle,inner sep = 0pt,minimum size =3pt,fill = black] (nx8y22) at (8,22) {};
\node[circle,inner sep = 0pt,minimum size =3pt,fill = black] (nx10y24) at (10,24) {};
\node[circle,inner sep = 0pt,minimum size =3pt,fill = black] (nx6y26) at (6,26) {};
\node[circle,inner sep = 0pt,minimum size =3pt,fill = black] (nx6y24) at (6,24) {};
\node[circle,inner sep = 0pt,minimum size =3pt,fill = black] (nx6y22) at (6,22) {};
\node[circle,inner sep = 0pt,minimum size =3pt,fill = black] (nx4y20) at (4,20) {};
\node[circle,inner sep = 0pt,minimum size =3pt,fill = black] (nx4y22) at (4,22) {};
\node[circle,inner sep = 0pt,minimum size =3pt,fill = black] (nx2y22) at (2,22) {};
\draw[->,black] (nx18y26) -- (nx20y26);
\draw[->,black] (nx20y26) -- (nx22y26);
\draw[->,black] (nx20y26) -- (nx20y22);
\draw[->,black] (nx22y22) -- (nx24y22);
\draw[->,black] (nx20y22) -- (nx22y22);
\draw[->,black] (nx24y20) -- (nx26y20);
\draw[->,black] (nx22y18) -- (nx24y18);
\draw[->,black] (nx22y22) -- (nx22y18);
\draw[->,black] (nx22y24) -- (nx22y22);
\draw[->,black] (nx24y24) -- (nx26y24);
\draw[->,black] (nx22y24) -- (nx24y24);
\draw[->,black] (nx22y26) -- (nx22y24);
\draw[->,black] (nx22y28) -- (nx22y26);
\draw[->,black] (nx24y28) -- (nx24y24);
\draw[->,black] (nx24y24) -- (nx24y22);
\draw[->,black] (nx24y22) -- (nx24y20);
\draw[->,black] (nx24y20) -- (nx24y18);
\draw[->,black] (nx24y18) -- (nx24y16);
\draw[->,black] (nx24y12) -- (nx24y10);
\draw[->,black] (nx22y12) -- (nx22y10);
\draw[->,black] (nx24y10) -- (nx26y10);
\draw[->,black] (nx22y10) -- (nx24y10);
\draw[->,black] (nx20y10) -- (nx22y10);
\draw[->,black] (nx26y10) -- (nx26y8);
\draw[->,black] (nx24y10) -- (nx24y8);
\draw[->,black] (nx24y8) -- (nx24y6);
\draw[->,black] (nx26y8) -- (nx28y8);
\draw[->,black] (nx24y8) -- (nx26y8);
\draw[->,black] (nx20y8) -- (nx24y8);
\draw[->,black] (nx10y8) -- (nx12y8);
\draw[->,black] (nx8y10) -- (nx8y8);
\draw[->,black] (nx6y10) -- (nx8y10);
\draw[->,black] (nx12y8) -- (nx14y8);
\draw[->,black] (nx12y10) -- (nx12y8);
\draw[->,black] (nx8y10) -- (nx12y10);
\draw[->,black] (nx8y12) -- (nx8y10);
\draw[->,black] (nx4y12) -- (nx8y12);
\draw[->,black] (nx10y8) -- (nx10y6);
\draw[->,black] (nx8y8) -- (nx10y8);
\draw[->,black] (nx6y8) -- (nx8y8);
\draw[->,black] (nx6y10) -- (nx6y8);
\draw[->,black] (nx4y10) -- (nx6y10);
\draw[->,black] (nx6y8) -- (nx6y6);
\draw[->,black] (nx2y8) -- (nx6y8);
\draw[->,black] (nx2y10) -- (nx4y10);
\draw[->,black] (nx4y12) -- (nx4y10);
\draw[->,black] (nx4y14) -- (nx4y12);
\draw[->,black] (nx12y20) -- (nx12y18);
\draw[->,black] (nx8y20) -- (nx8y18);
\draw[->,black] (nx12y20) -- (nx14y20);
\draw[->,black] (nx10y20) -- (nx12y20);
\draw[->,black] (nx8y20) -- (nx10y20);
\draw[->,black] (nx4y20) -- (nx8y20);
\draw[->,black] (nx10y22) -- (nx10y20);
\draw[->,black] (nx8y22) -- (nx8y20);
\draw[->,black] (nx4y22) -- (nx4y20);
\draw[->,black] (nx10y22) -- (nx14y22);
\draw[->,black] (nx8y22) -- (nx10y22);
\draw[->,black] (nx6y22) -- (nx8y22);
\draw[->,black] (nx4y22) -- (nx6y22);
\draw[->,black] (nx2y22) -- (nx4y22);
\draw[->,black] (nx10y24) -- (nx10y22);
\draw[->,black] (nx6y24) -- (nx10y24);
\draw[->,black] (nx6y24) -- (nx6y22);
\draw[->,black] (nx6y26) -- (nx6y24);
\draw (3,6) node {$\textrm{TT}$};
\draw (3,25) node {$\textrm{IF}$};
\draw (26,27) node {$\textrm{LL}$};
\draw (27,6) node {$\textrm{VV}$};
\draw (28,8) node[anchor = west] {$b$};
\draw (24,6) node[anchor = north] {$a$};
\draw (20,8) node[anchor = east] {$b_1$};
\draw (20,10) node[anchor = east] {$b_2$};
\draw (24,12) node[anchor = south] {$a_2$};
\draw (22,12) node[anchor = south] {$a_1$};
\draw (2,8) node[anchor = east] {$b_2$};
\draw (2,10) node[anchor = east] {$b_1$};
\draw (4,14) node[anchor = south] {$a$};
\draw (6,6) node[anchor = west] {$a_2$};
\draw (10,6) node[anchor = west] {$a_1$};
\draw (14,8) node[anchor = north] {$b$};
\draw (24,16) node[anchor = north] {$b$};
\draw (26,20) node[anchor = west] {$a_1$};
\draw (26,24) node[anchor = west] {$a_2$};
\draw (14,22) node[anchor = west] {$b_2$};
\draw (14,20) node[anchor = west] {$b_1$};
\draw (12,18) node[anchor = north] {$a_2$};
\draw (8,18) node[anchor = north] {$a_1$};
\draw (24,28) node[anchor = south] {$b_2$};
\draw (22,28) node[anchor = south] {$b_1$};
\draw (18,26) node[anchor = east] {$a$};
\draw (6,26) node[anchor = south] {$a$};
\draw (2,22) node[anchor = east] {$b$};
\end{tikzpicture}
\legende{Four special subgraphs, from top to bottom and left to right : $\textrm{IF}$, $\textrm{LL}$, $\textrm{TT}$, $\textrm{VV}$. First three have the property that if there are two arc-disjoint paths, one between $a$ and $a_i$ and the other between $b$ and $b_j$, then $i=j$.}
\label{figrouteurs}
\end{figure}

\begin{myclaim}\label{claimrouteur}
Let $P_1$ be a path between $a$ and one of $a_i$, $i \in \{1,2\}$, and $P_2$ be a path between $b$ and $b_j$, $j \in \{1,2\}$, in $\textrm{IF}$, $\textrm{LL}$ or $\textrm{TT}$. If $P_1$ and $P_2$ are arc-disjoint, then $i=j$.\qed\\
\end{myclaim}

\begin{myclaim}\label{claimverifieur}
There are not two arc-disjoint paths in $\textrm{VV}$, one from $b_2$ to $b$ and the other from $a_1$ to $a$.\qed
\end{myclaim}

Claim~\ref{claimrouteur} and Claim~\ref{claimverifieur} can be readily checked. We will also need the graph $\textrm{ON}$ introduced in Figure~\ref{figon}. We now describe the full graph for the reduction. $G$ is build from $G_1$ in the model of Figure~\ref{figbigfig}. $G$ is built from a grid with $2p+n$ columns and $2p$ rows. The subgrid defined by columns $p+1$ to $p+n$ and rows $1$ to $p$ is $G_1$. Note that two rows in $G_1$ correspond to one row in $G$. Squares $G(i,i)$, $G(p+1-i, n+p+i)$, $G(2p+1-i,i)$ and $G(p+i, n+p+i)$, for all $i \in \ival{1}{p}$, are special graph $\textrm{IF}$, $\textrm{TT}$, $\textrm{LL}$ and $\textrm{VV}$ respectively. Others are either $\textrm{NO}$ or $\textrm{ON}$, according to the figure.  Rows, columns, vertical cuts and horizontal cuts are defined in the same way as in Section~\ref{sec:reduction}. We add four terminals, one for each side of the grid (see the figure).

\begin{figure}
\centering
\begin{tikzpicture}[x=0.4cm,y=0.4cm,>=latex]
% Fichier d'origine : biflot-on.pgf
\node[circle,inner sep = 0pt,minimum size =3pt,fill = black] (nx5y21) at (5,21) {};
\node[circle,inner sep = 0pt,minimum size =3pt,fill = black] (nx5y23) at (5,23) {};
\node[circle,inner sep = 0pt,minimum size =3pt,fill = black] (nx9y23) at (9,23) {};
\node[circle,inner sep = 0pt,minimum size =3pt,fill = black] (nx9y21) at (9,21) {};
\node[circle,inner sep = 0pt,minimum size =3pt,fill = black] (nx7y19) at (7,19) {};
\node[circle,inner sep = 0pt,minimum size =3pt,fill = black] (nx7y21) at (7,21) {};
\node[circle,inner sep = 0pt,minimum size =3pt,fill = black] (nx7y23) at (7,23) {};
\node[circle,inner sep = 0pt,minimum size =3pt,fill = black] (nx7y25) at (7,25) {};
\draw[->,black] (nx7y21) -- (nx9y21);
\draw[->,black] (nx7y23) -- (nx9y23);
\draw[->,black] (nx5y23) -- (nx7y23);
\draw[->,black] (nx5y21) -- (nx7y21);
\draw[->,black] (nx7y21) -- (nx7y19);
\draw[->,black] (nx7y23) -- (nx7y21);
\draw[->,black] (nx7y25) -- (nx7y23);
\draw (7,19) node[anchor = north] {$a'$};
\draw (9,21) node[anchor = west] {$c'$};
\draw (9,23) node[anchor = west] {$b'$};
\draw (5,21) node[anchor = east] {$c$};
\draw (5,23) node[anchor = east] {$b$};
\draw (7,25) node[anchor = south] {$a$};
\end{tikzpicture}
\legende{Graph $\textrm{ON}$, it has the same property as graph $\textrm{NO}$.}
\label{figon}
\end{figure}

\begin{figure}
\begin{center}
\begin{tikzpicture}[x=0.25cm,y=0.25cm,>=latex]
\fill[black, nearly transparent,rounded corners=2pt] (40.,6.) rectangle (43.,9.);
\draw (41.5,7.5) node {$\textrm{NO}$};
\fill[black, nearly transparent,rounded corners=2pt] (20.,6.) rectangle (23.,9.);
\draw (21.5,7.5) node {$\textrm{ON}$};
\fill[black, nearly transparent,rounded corners=2pt] (36.,22.) rectangle (39.,25.);
\draw (37.5,23.5) node {$\textrm{ON}$};
\fill[black, nearly transparent,rounded corners=2pt] (32.,22.) rectangle (35.,25.);
\draw (33.5,23.5) node {$\textrm{ON}$};
\fill[black, nearly transparent,rounded corners=2pt] (32.,18.) rectangle (35.,21.);
\draw (33.5,19.5) node {$\textrm{ON}$};
\fill[black, nearly transparent,rounded corners=2pt] (40.,18.) rectangle (43.,21.);
\draw (41.5,19.5) node {$\textrm{NO}$};
\fill[black, nearly transparent,rounded corners=2pt] (40.,14.) rectangle (43.,17.);
\draw (41.5,15.5) node {$\textrm{NO}$};
\fill[black, nearly transparent,rounded corners=2pt] (36.,14.) rectangle (39.,17.);
\draw (37.5,15.5) node {$\textrm{NO}$};
\fill[black, nearly transparent,rounded corners=2pt] (36.,10.) rectangle (39.,13.);
\draw (37.5,11.5) node {$\textrm{NO}$};
\fill[black, nearly transparent,rounded corners=2pt] (40.,10.) rectangle (43.,13.);
\draw (41.5,11.5) node {$\textrm{NO}$};
\fill[black, nearly transparent,rounded corners=2pt] (36.,2.) rectangle (39.,5.);
\draw (37.5,3.5) node {$\textrm{ON}$};
\fill[black, nearly transparent,rounded corners=2pt] (32.,2.) rectangle (35.,5.);
\draw (33.5,3.5) node {$\textrm{ON}$};
\fill[black, nearly transparent,rounded corners=2pt] (28.,2.) rectangle (31.,5.);
\draw (29.5,3.5) node {$\textrm{ON}$};
\fill[black, nearly transparent,rounded corners=2pt] (24.,2.) rectangle (27.,5.);
\draw (25.5,3.5) node {$\textrm{ON}$};
\fill[black, nearly transparent,rounded corners=2pt] (20.,2.) rectangle (23.,5.);
\draw (21.5,3.5) node {$\textrm{ON}$};
\fill[black, nearly transparent,rounded corners=2pt] (16.,2.) rectangle (19.,5.);
\draw (17.5,3.5) node {$\textrm{ON}$};
\fill[black, nearly transparent,rounded corners=2pt] (12.,2.) rectangle (15.,5.);
\draw (13.5,3.5) node {$\textrm{ON}$};
\fill[black, nearly transparent,rounded corners=2pt] (32.,6.) rectangle (35.,9.);
\draw (33.5,7.5) node {$\textrm{ON}$};
\fill[black, nearly transparent,rounded corners=2pt] (28.,6.) rectangle (31.,9.);
\draw (29.5,7.5) node {$\textrm{ON}$};
\fill[black, nearly transparent,rounded corners=2pt] (24.,6.) rectangle (27.,9.);
\draw (25.5,7.5) node {$\textrm{ON}$};
\fill[black, nearly transparent,rounded corners=2pt] (16.,6.) rectangle (19.,9.);
\draw (17.5,7.5) node {$\textrm{ON}$};
\fill[black, nearly transparent,rounded corners=2pt] (28.,10.) rectangle (31.,13.);
\draw (29.5,11.5) node {$\textrm{ON}$};
\fill[black, nearly transparent,rounded corners=2pt] (24.,10.) rectangle (27.,13.);
\draw (25.5,11.5) node {$\textrm{ON}$};
\fill[black, nearly transparent,rounded corners=2pt] (20.,10.) rectangle (23.,13.);
\draw (21.5,11.5) node {$\textrm{ON}$};
\fill[black, nearly transparent,rounded corners=2pt] (8.,6.) rectangle (11.,9.);
\draw (9.5,7.5) node {$\textrm{NO}$};
\fill[black, nearly transparent,rounded corners=2pt] (12.,10.) rectangle (15.,13.);
\draw (13.5,11.5) node {$\textrm{NO}$};
\fill[black, nearly transparent,rounded corners=2pt] (8.,10.) rectangle (11.,13.);
\draw (9.5,11.5) node {$\textrm{NO}$};
\fill[black, nearly transparent,rounded corners=2pt] (8.,18.) rectangle (11.,21.);
\draw (9.5,19.5) node {$\textrm{NO}$};
\fill[black, nearly transparent,rounded corners=2pt] (8.,14.) rectangle (11.,17.);
\draw (9.5,15.5) node {$\textrm{NO}$};
\fill[black, nearly transparent,rounded corners=2pt] (12.,14.) rectangle (15.,17.);
\draw (13.5,15.5) node {$\textrm{NO}$};
\fill[black, nearly transparent,rounded corners=2pt] (16.,18.) rectangle (19.,21.);
\draw (17.5,19.5) node {$\textrm{ON}$};
\fill[black, nearly transparent,rounded corners=2pt] (16.,22.) rectangle (19.,25.);
\draw (17.5,23.5) node {$\textrm{ON}$};
\fill[black, nearly transparent,rounded corners=2pt] (12.,22.) rectangle (15.,25.);
\draw (13.5,23.5) node {$\textrm{ON}$};
\fill[black, semitransparent,rounded corners=2pt] (40.,2.) rectangle (43.,5.);
\draw (41.5,3.5) node {$\textrm{VV}$};
\fill[black, semitransparent,rounded corners=2pt] (36.,6.) rectangle (39.,9.);
\draw (37.5,7.5) node {$\textrm{VV}$};
\fill[black, semitransparent,rounded corners=2pt] (32.,10.) rectangle (35.,13.);
\draw (33.5,11.5) node {$\textrm{VV}$};
\fill[black, semitransparent,rounded corners=2pt] (8.,2.) rectangle (11.,5.);
\draw (9.5,3.5) node {$\textrm{LL}$};
\fill[black, semitransparent,rounded corners=2pt] (12.,6.) rectangle (15.,9.);
\draw (13.5,7.5) node {$\textrm{LL}$};
\fill[black, semitransparent,rounded corners=2pt] (16.,10.) rectangle (19.,13.);
\draw (17.5,11.5) node {$\textrm{LL}$};
\fill[black, semitransparent,rounded corners=2pt] (40.,22.) rectangle (43.,25.);
\draw (41.5,23.5) node {$\textrm{TT}$};
\fill[black, semitransparent,rounded corners=2pt] (36.,18.) rectangle (39.,21.);
\draw (37.5,19.5) node {$\textrm{TT}$};
\fill[black, semitransparent,rounded corners=2pt] (32.,14.) rectangle (35.,17.);
\draw (33.5,15.5) node {$\textrm{TT}$};
\fill[black, semitransparent,rounded corners=2pt] (20.,14.) rectangle (31.,25.);
\fill[black, semitransparent,rounded corners=2pt] (8.,22.) rectangle (11.,25.);
\draw (9.5,23.5) node {$\textrm{IF}$};
\fill[black, semitransparent,rounded corners=2pt] (12.,18.) rectangle (15.,21.);
\draw (13.5,19.5) node {$\textrm{IF}$};
\fill[black, semitransparent,rounded corners=2pt] (16.,14.) rectangle (19.,17.);
\draw (17.5,15.5) node {$\textrm{IF}$};
\draw[->,black,thick] (25,30) -- (41,26);
\draw[->,black,thick] (25,30) -- (37,26);
\draw[->,black,thick] (25,30) -- (33,26);
\draw[->,black,thick] (25,30) -- (29,26);
\draw[->,black,thick] (25,30) -- (25,26);
\draw[->,black,thick] (25,30) -- (9,26);
\draw[->,black,thick] (25,30) -- (21,26);
\draw[->,black,thick] (25,30) -- (17,26);
\draw[->,black,thick] (25,30) -- (13,26);
\draw[->,black,thick] (42,14) -- (42,13);
\draw[->,black,thick] (41,14) -- (41,13);
\draw[->,black,thick] (38,14) -- (38,13);
\draw[->,black,thick] (37,14) -- (37,13);
\draw[->,black,thick] (42,18) -- (42,17);
\draw[->,black,thick] (41,18) -- (41,17);
\draw[->,black,thick] (38,18) -- (38,17);
\draw[->,black,thick] (37,18) -- (37,17);
\draw[->,black,thick] (33,18) -- (33,17);
\draw[->,black,thick] (33,22) -- (33,21);
\draw[->,black,thick] (37,22) -- (37,21);
\draw[->,black,thick] (42,22) -- (42,21);
\draw[->,black,thick] (41,22) -- (41,21);
\draw[->,black,thick] (41,26) -- (41,25);
\draw[->,black,thick] (37,26) -- (37,25);
\draw[->,black,thick] (33,26) -- (33,25);
\draw[->,black,thick] (29,26) -- (29,25);
\draw[->,black,thick] (25,26) -- (25,25);
\draw[->,black,thick] (21,26) -- (21,25);
\draw[->,black,thick] (13,26) -- (13,25);
\draw[->,black,thick] (17,26) -- (17,25);
\draw[->,black,thick] (17,22) -- (17,21);
\draw[->,black,thick] (13,22) -- (13,21);
\draw[->,black,thick] (9,2) -- (9,1);
\draw[->,black,thick] (17,18) -- (17,17);
\draw[->,black,thick] (29,14) -- (29,13);
\draw[->,black,thick] (25,14) -- (25,13);
\draw[->,black,thick] (21,14) -- (21,13);
\draw[->,black,thick] (42,6) -- (42,5);
\draw[->,black,thick] (41,6) -- (41,5);
\draw[->,black,thick] (42,10) -- (42,9);
\draw[->,black,thick] (41,10) -- (41,9);
\draw[->,black,thick] (38,10) -- (38,9);
\draw[->,black,thick] (37,10) -- (37,9);
\draw[->,black,thick] (33,10) -- (33,9);
\draw[->,black,thick] (29,10) -- (29,9);
\draw[->,black,thick] (25,10) -- (25,9);
\draw[->,black,thick] (21,10) -- (21,9);
\draw[->,black,thick] (17,10) -- (17,9);
\draw[->,black,thick] (9,1) -- (25,0);
\draw[->,black,thick] (13,1) -- (25,0);
\draw[->,black,thick] (17,1) -- (25,0);
\draw[->,black,thick] (21,1) -- (25,0);
\draw[->,black,thick] (25,1) -- (25,0);
\draw[->,black,thick] (29,1) -- (25,0);
\draw[->,black,thick] (33,1) -- (25,0);
\draw[->,black,thick] (37,1) -- (25,0);
\draw[->,black,thick] (41,1) -- (25,0);
\draw[->,black,thick] (41,2) -- (41,1);
\draw[->,black,thick] (37,6) -- (37,5);
\draw[->,black,thick] (37,2) -- (37,1);
\draw[->,black,thick] (33,2) -- (33,1);
\draw[->,black,thick] (33,6) -- (33,5);
\draw[->,black,thick] (29,2) -- (29,1);
\draw[->,black,thick] (29,6) -- (29,5);
\draw[->,black,thick] (25,6) -- (25,5);
\draw[->,black,thick] (25,2) -- (25,1);
\draw[->,black,thick] (21,2) -- (21,1);
\draw[->,black,thick] (21,6) -- (21,5);
\draw[->,black,thick] (17,6) -- (17,5);
\draw[->,black,thick] (17,2) -- (17,1);
\draw[->,black,thick] (13,2) -- (13,1);
\draw[->,black,thick] (13,6) -- (13,5);
\draw[->,black,thick] (14,10) -- (14,9);
\draw[->,black,thick] (13,10) -- (13,9);
\draw[->,black,thick] (14,14) -- (14,13);
\draw[->,black,thick] (13,14) -- (13,13);
\draw[->,black,thick] (14,18) -- (14,17);
\draw[->,black,thick] (13,18) -- (13,17);
\draw[->,black,thick] (10,6) -- (10,5);
\draw[->,black,thick] (9,6) -- (9,5);
\draw[->,black,thick] (10,10) -- (10,9);
\draw[->,black,thick] (9,10) -- (9,9);
\draw[->,black,thick] (10,14) -- (10,13);
\draw[->,black,thick] (9,14) -- (9,13);
\draw[->,black,thick] (10,18) -- (10,17);
\draw[->,black,thick] (9,18) -- (9,17);
\draw[->,black,thick] (10,22) -- (10,21);
\draw[->,black,thick] (9,22) -- (9,21);
\draw[->,black,thick] (9,26) -- (9,25);
\draw[->,black,thick] (44,4) -- (48,14);
\draw[->,black,thick] (43,4) -- (44,4);
\draw[->,black,thick] (39,3) -- (40,3);
\draw[->,black,thick] (39,4) -- (40,4);
\draw[->,black,thick] (35,3) -- (36,3);
\draw[->,black,thick] (35,4) -- (36,4);
\draw[->,black,thick] (31,3) -- (32,3);
\draw[->,black,thick] (31,4) -- (32,4);
\draw[->,black,thick] (27,3) -- (28,3);
\draw[->,black,thick] (27,4) -- (28,4);
\draw[->,black,thick] (23,3) -- (24,3);
\draw[->,black,thick] (23,4) -- (24,4);
\draw[->,black,thick] (19,3) -- (20,3);
\draw[->,black,thick] (19,4) -- (20,4);
\draw[->,black,thick] (15,3) -- (16,3);
\draw[->,black,thick] (15,4) -- (16,4);
\draw[->,black,thick] (11,3) -- (12,3);
\draw[->,black,thick] (11,4) -- (12,4);
\draw[->,black,thick] (44,8) -- (48,14);
\draw[->,black,thick] (43,8) -- (44,8);
\draw[->,black,thick] (39,8) -- (40,8);
\draw[->,black,thick] (35,7) -- (36,7);
\draw[->,black,thick] (35,8) -- (36,8);
\draw[->,black,thick] (31,7) -- (32,7);
\draw[->,black,thick] (31,8) -- (32,8);
\draw[->,black,thick] (27,7) -- (28,7);
\draw[->,black,thick] (27,8) -- (28,8);
\draw[->,black,thick] (23,7) -- (24,7);
\draw[->,black,thick] (23,8) -- (24,8);
\draw[->,black,thick] (19,7) -- (20,7);
\draw[->,black,thick] (19,8) -- (20,8);
\draw[->,black,thick] (15,7) -- (16,7);
\draw[->,black,thick] (15,8) -- (16,8);
\draw[->,black,thick] (11,8) -- (12,8);
\draw[->,black,thick] (44,12) -- (48,14);
\draw[->,black,thick] (43,12) -- (44,12);
\draw[->,black,thick] (39,12) -- (40,12);
\draw[->,black,thick] (35,12) -- (36,12);
\draw[->,black,thick] (15,12) -- (16,12);
\draw[->,black,thick] (11,12) -- (12,12);
\draw[->,black,thick] (44,16) -- (48,14);
\draw[->,black,thick] (43,16) -- (44,16);
\draw[->,black,thick] (39,16) -- (40,16);
\draw[->,black,thick] (35,16) -- (36,16);
\draw[->,black,thick] (15,16) -- (16,16);
\draw[->,black,thick] (11,16) -- (12,16);
\draw[->,black,thick] (44,20) -- (48,14);
\draw[->,black,thick] (43,20) -- (44,20);
\draw[->,black,thick] (39,20) -- (40,20);
\draw[->,black,thick] (35,19) -- (36,19);
\draw[->,black,thick] (35,20) -- (36,20);
\draw[->,black,thick] (31,19) -- (32,19);
\draw[->,black,thick] (31,20) -- (32,20);
\draw[->,black,thick] (19,19) -- (20,19);
\draw[->,black,thick] (19,20) -- (20,20);
\draw[->,black,thick] (15,19) -- (16,19);
\draw[->,black,thick] (15,20) -- (16,20);
\draw[->,black,thick] (11,20) -- (12,20);
\draw[->,black,thick] (44,24) -- (48,14);
\draw[->,black,thick] (43,24) -- (44,24);
\draw[->,black,thick] (39,23) -- (40,23);
\draw[->,black,thick] (39,24) -- (40,24);
\draw[->,black,thick] (35,23) -- (36,23);
\draw[->,black,thick] (35,24) -- (36,24);
\draw[->,black,thick] (31,23) -- (32,23);
\draw[->,black,thick] (31,24) -- (32,24);
\draw[->,black,thick] (19,23) -- (20,23);
\draw[->,black,thick] (19,24) -- (20,24);
\draw[->,black,thick] (15,23) -- (16,23);
\draw[->,black,thick] (15,24) -- (16,24);
\draw[->,black,thick] (11,23) -- (12,23);
\draw[->,black,thick] (11,24) -- (12,24);
\draw[->,black,thick] (7,24) -- (8,24);
\draw[->,black,thick] (7,20) -- (8,20);
\draw[->,black,thick] (7,16) -- (8,16);
\draw[->,black,thick] (7,12) -- (8,12);
\draw[->,black,thick] (7,8) -- (8,8);
\draw[->,black,thick] (7,4) -- (8,4);
\draw[->,black,thick] (3,14) -- (7,4);
\draw[->,black,thick] (3,14) -- (7,8);
\draw[->,black,thick] (3,14) -- (7,12);
\draw[->,black,thick] (3,14) -- (7,16);
\draw[->,black,thick] (3,14) -- (7,20);
\draw[->,black,thick] (3,14) -- (7,24);

\draw[->,black,thick] (17,14) -- (17,13);
\draw[->,black,thick] (19,15) -- (20,15);
\draw[->,black,thick] (31,11) -- (32,11);
\draw[->,black,thick] (27,11) -- (28,11);
\draw[->,black,thick] (23,11) -- (24,11);
\draw[->,black,thick] (19,11) -- (20,11);

\draw[->,black,thick] (18,14) -- (18,13);
\draw[->,black,thick] (19,16) -- (20,16);
\draw[->,black,thick] (31,12) -- (32,12);
\draw[->,black,thick] (27,12) -- (28,12);
\draw[->,black,thick] (23,12) -- (24,12);
\draw[->,black,thick] (19,12) -- (20,12);

\draw[->,black,thick] (31,16) -- (32,16);
\draw[->,black,thick] (34,14) -- (34,13);

\draw[->,black,thick] (31,15) -- (32,15);
\draw[->,black,thick] (33,14) -- (33,13);

\draw (25,0) node[anchor = north] {$t_2$};
\draw (25,30) node[anchor = south] {$t_1$};
\draw (48,14) node[anchor = west] {$s_2$};
\draw (3,14) node[anchor = east] {$s_1$};
\draw (25,19) node {$G_1$};
\end{tikzpicture}
\end{center}
\legende{The graph for the reduction with a formula containing three clauses over three variables.}
\label{figbigfig}
\end{figure}
We add a demand of $2p$ from $s_1$ to $s_2$, and $2p+n$ from $t_1$ to $t_2$.

\begin{myclaim}
$G$ is acyclic.
\end{myclaim}

\begin{proof}
Observe that in the grid, all arcs are from left to right or from top to bottom, and the special graphs are all acyclic.
\end{proof}

\begin{myclaim}
There is exactly one path going from top to bottom in each column of $G$, this path never leaves the column. There is exactly one path going from left to right in each row of $G$, this path never leaves the row.
\end{myclaim}

\begin{proof}
Because vertical and horizontal cuts are directed, and $\{s_1\}$, $\{s_2\}$, $\{t_1\}$, $\{t_2\}$ are tight cuts.
\end{proof}

\begin{myclaim}\label{claimmain}
If a horizontal path leaves $G_1$ by the lower edge of its row in the $(n+p)^{\textrm{th}}$ vertical cut, then the same path enters $G_1$ by the lower edge of its row in the $p^{\textrm{th}}$ vertical cut.
\end{myclaim}

\begin{proof}
In each square $\textrm{ON}$ or $\textrm{NO}$, there is exactly one path from $b$ to $b'$ or from $c$ to $c'$, and one path from $a$ to $a'$, because of properties of these gadgets. Consequently, a vertical path leaves an $\textrm{IF}$ gadget by the left if and only if it enters the $\textrm{LL}$ gadget of the same column by the left. Similarly for paths between $\textrm{LL}$ and $\textrm{VV}$, for paths between $\textrm{TT}$ and $\textrm{VV}$, for paths between $\textrm{IF}$ and $G_1$ and for paths between $G_1$ and $\textrm{TT}$. If a path leaves $G_1$ by the lower edges of its row, say row $i \in \ival{1}{p}$, then it enters $G(i,2p+n+1-i)$ by vertex $b_2$. The vertical path of column $2p+n+1-i$ leaves $G(i,2p+n+1-i)$ by vertex $a_2$ by Claim~\ref{claimrouteur}, and then enters $G(2p+1-i, 2p+n+1-i)$ by vertex $a_1$. Thus the horizontal path in row $2p+1-i$ goes in $G(2p+1-i, 2p+n+1-i)$ using vertex $b_1$ by Claim~\ref{claimverifieur}, and leaves $G(2p+1-i,i)$ by vertex $a_1$. Using Claim~\ref{claimrouteur}, the vertical path in column $i$ goes in $G(2p+1-i,i)$ by vertex $b_1$, thus goes out $G(i,i)$ by vertex $a_1$. By Claim~\ref{claimrouteur} again, the horizontal path of row $i$ leaves $G(i,i)$ by $b_1$, and then enters $G_1$ by the lower edge, proving the claim.
\end{proof}

\begin{proof}(of Theorem~\ref{th:diflot})
Claim~\ref{claimmain} proves that the path in $G_1$ satisfies the condition $(ii)$ of Claim~\ref{claimfirst}: if there is a solution to the arc disjoint path problem, the formula is satisfiable. The converse is also true, it is sufficient to extend the solution for $G_1$ naturally.\\
As the construction is obviously polynomial, we found a Karp reduction between the two problems. The arc-disjoint paths problem being in NP, Theorem~\ref{th:diflot} is proved. Note that we can identify $t_1$ with $s_2$, and $t_2$ with $s_1$, thus proving that it is still true with $G+H$ planar and only two terminals.
\end{proof}

\begin{corollary}
The arc-disjoint paths problem in planar graphs is NP-complete, even if the demand graph has only two arcs, with one of request $1$ (one flow plus one path).
\end{corollary}

\begin{proof}
We modify the preceding reduction. We remove $t_1$ and $t_2$, and add arcs from the bottom of a column, to the top of the next column to the left. This preserves planarity. Then we add a demand arc from the bottom of the leftmost column to the top of the rightmost column, with demand $1$. We keep the arc $s_2s_1$. The new demand must be routed through the new arcs because of the vertical tight cuts. Thus, this transformation preserves the property of the original reduction.
\end{proof}

\bigskip

The problem where the total amount of demands is fixed remains open in digraphs, and in particular, the special case where we want to find a cycle in a planar digraph, that goes through two specified vertices. This last problem is mentioned in~\cite{Schrijverbook}.

\end{document}